\documentclass[11pt,a4paper]{article}
\usepackage{fullpage}

\usepackage{amsmath,amsthm,amssymb,hyperref,verbatim}
\usepackage{tikz}
\usetikzlibrary{calc,patterns,decorations.pathreplacing}

\usepackage{natbib}
\allowdisplaybreaks
\usepackage{hyperref}
\hypersetup{
     colorlinks   = true ,% Colours links instead of boxes
     urlcolor     = blue,
     linkcolor    = blue,
     citecolor    = blue
}

\DeclareMathOperator*{\argmax}{arg\,max}
\DeclareMathOperator*{\err}{err}
\DeclareMathOperator*{\N}{\mathbb{N}}
\DeclareMathOperator*{\R}{\mathbb{R}}
\DeclareMathOperator*{\one}{\mathbb{I}}

\theoremstyle{plain}
\newtheorem{thm}{Theorem}
\newtheorem{lemma}[thm]{Lemma}
\newtheorem{cor}[thm]{Corollary}
\newtheorem{defn}[thm]{Definition}
\newtheorem{claim}[thm]{Claim}
\newtheorem{obs}[thm]{Observation}

\newcommand{\ABCS}{{\sc TargetABCS}}
\newcommand{\SeqThiele}{{\sc TargetSeqThiele}}
\newcommand{\score}{\text{sc}}
\newcommand{\ccrule}{f_{\mathit{CC}}}

\usepackage{xcolor, soul}

\def\DEBUG{true}

\ifdefined\DEBUG
\fi

\title{\bf The Complexity of Learning Approval-Based Multiwinner Voting Rules\thanks{A preliminary version of the results in this paper appeared in {\em Proceedings of the 36th AAAI Conference on Artificial Intelligence} (AAAI 2022), pages 4235--4932, 2022.}}
\author{Ioannis Caragiannis\thanks{Department of Computer Science, Aarhus University, {\AA}bogade 34, 8200 Aarhus N, Denmark. Email: \{\href{mailto:iannis@cs.au.dk}{iannis}, \href{mailto:karl@cs.au.dk}{karl}\}@cs.au.dk.} \and Karl Fehrs\footnotemark[2]}
\date{}

\begin{document}

\maketitle

\begin{abstract}
We study the {PAC} learnability of multiwinner voting, focusing on the class of approval-based committee scoring (ABCS) rules. These are voting rules applied on profiles with approval ballots, where each voter approves some of the candidates. According to ABCS rules, each committee of $k$ candidates collects from each voter a score, which depends on the size of the voter's ballot and on the size of its intersection with the committee. Then, committees of maximum score are the winning ones. Our goal is to learn a target rule (i.e., to learn the corresponding scoring function) using information about the winning committees of a small number of sampled profiles. Despite the existence of exponentially many outcomes compared to single-winner elections, we show that the sample complexity is still low: a polynomial number of samples carries enough information for learning the target rule with high confidence and accuracy. Unfortunately, even simple tasks that need to be solved for learning from these samples are intractable. We prove that deciding whether there exists some ABCS rule that makes a given committee winning in a given profile is a computationally hard problem. Our results extend to the class of sequential Thiele rules, which have received attention recently due to their simplicity.
\end{abstract}

\section{Introduction}
Voting has been used for centuries to aggregate individual preferences into a common decision. In addition to its traditional use for electing governments or for decision making in management boards, it has also been proved useful in novel applications where individual ratings need to be summarized as collective knowledge. But, is there a general recipe on how preferences should be aggregated? Fortunately, there is no ``golden'' voting rule and this has led social choice theory ---and, in particular, its modern computational branch~\citep{BCEL+16}--- onto exciting research endeavours.

A popular approach has aimed, quite successfully, to evaluate voting rules in terms of desirable axioms they must satisfy. Well-known impossibilities, e.g., see~\citet{A51}, showcase the limitations of this approach. Deviating from this {\em axiomatic} treatment, recent works view voting rules as optimized decision making methods, perhaps tailored to particular applications. In this context, the {\em data-driven design} of voting rules is a very natural approach. The goal is to derive a voting rule from a set of known preferences which are accompanied by favored outcomes under these preferences. The hope then is for the resulting rule to be equally well-suited to more general preferences where the favored outcomes are unknown. The current paper aims to study the potentials and limitations of this approach.

We focus on {\em multiwinner} voting rules~\citep{FSST17a}, which on input the preferences of $n$ voters over $m$ available candidates, return as outcome one or more committees of candidates of fixed size $k$. In particular, we study {\em approval-based} voting~\citep{LS10,LS23}, where the preference of a voter is simply the set of candidates she approves. And, more concretely, we consider the class of {\em approval-based committee scoring} ({\em ABCS}) rules, defined by~\citet{LS21}. An ABCS rule follows a common format. It employs a scoring function, according to which each voter awards a score to each committee of $k$ candidates. This score depends on the ballot size (the number of  candidates the voter approves) and the size of its intersection with the committee. Different scoring functions can be used to define different voting rules.

A natural application area of approval-based voting are rating tasks. For example, consider a website specialized on cultural events that aims to present every week the top-20 performances in the theaters of a city. A simple way to compute this is to ask website visitors for their opinions and aggregate them to a top-20 list. Approval-based multiwinner voting can be used here as follows.
\begin{itemize}
\item Each participant is asked for her favorite performances (i.e., for her approval vote) among the ones available (i.e., among the alternatives).
\item Then, an ABCS rule can be used to compute the top-20 (i.e., the winning committee).
\end{itemize}

Deciding on the best ABCS rule depends on the application at hand. For example, under a rule that favours {\em individual excellence}, each voter assigns to each committee a score that is equal to the number of candidates in the committee the voter approves. Another rule could give just one point to each committee that has a non-empty intersection with the voter's ballot; such a rule would promote {\em representation of voters}. In practice, situations with such a one-dimensional objective for a voting rule are extremely rare. This issue has been extensively studied in the literature, e.g., see the work of \citet{FSST17b,FT18,LS20,JS22}. At the same time, hand-picking an ABCS rule that satisfies multiple objectives (and even accounting for potential trade-offs between objectives) might prove difficult. Instead, it may be easier to derive the characteristics of the desired rule from data when suitable data is available or when it is possible to generate such data. In this paper, we assume the availability of data in the form of preference profiles and corresponding desired winning committees.

Consider the above example of rating cultural events.  A good ABCS rule for picking the top-20 theater performances may aim at ensuring that the choice is representative of the website visitors' common interests while also taking into account fringe works that nevertheless appear highly outstanding to a few of the theater-goers. A data-driven approach to decide such a rule could be implemented by the website operators as follows.
\begin{itemize}
\item For the first ten weeks of operations, the website collects the input from the visitors but uses a small set of (expensive) experts every week to decide an ``ideal'' top-20.
\item Then, these ten pairs of visitor input and expert top-20 constitute a set of profile-winning committee examples.
\item An ABCS rule derived from these data is applied in the subsequent weeks to visitors' input to simulate (in an inexpensive way) the expert opinion.
\end{itemize}
This indicative scenario involves several thousands of voters, more than $100$ alternatives, and winning committees of size $20$. Other rating applications, e.g., for hotels, restaurants, or local businesses, or platforms for evaluating business proposals, investment opportunities, and microloan applications could benefit from a similar approach. Arguably, the best rule for the  application at hand should at least agree with given profile-winning committee data, and, ideally, produce desirable outcomes for unknown preference profiles. Can such a data-driven selection of an ABCS rule be effective?

We explore this question using the {\em PAC} (probably approximately correct) {\em learning} framework. We follow a similar methodological approach with~\citet{PZPR09}, who addressed the same question for single-winner voting rules. In the terminology of PAC learning, we would like to determine the {\em sample complexity} of the class of ABCS rules. How many samples (profiles and corresponding winning committees) are necessary and sufficient so that an ABCS rule that agrees with these data points can be learnt? However, the answer to this question addresses our challenge only partially. Indeed, low sample complexity does not necessarily imply {\em efficient} learning, as the computational problem of finding an ABCS rule that fits the given data can be hard.

\subsection{Our contribution and techniques}\label{sec:our-contribution}
Our first result states that the class of ABCS rules has only polynomial sample complexity (Section~\ref{sec:learnability}). Using a variant of the {\em multiclass fundamental theorem} in PAC learning (Theorem~\ref{thm:sample-vs-graph_Natarajan}), we obtain our sample complexity bounds by proving upper bounds on the {\em graph dimension} and lower bounds on the {\em Natarajan dimension} of the class of ABCS rules. For our upper bound,  we establish a connection between the graph dimension and the number of different sign patterns of a set of linear functions. Then, a result in algebraic combinatorics ---originally proved by~\citet{W68} and later refined by~\citet{A96}--- is used to upper-bound this number of sign patterns and, consequently, the graph dimension and the sample complexity of the class of ABCS rules. 

On the negative side, we give strong evidence that efficient PAC learnability of ABCS rules is not possible. We show that given a profile of approval votes and a committee, deciding whether there is an ABCS rule that makes this committee winning is a $\mathsf{coW}[1]$-hard\footnote{We follow standard notions from parameterized complexity theory, such as $\mathsf{W}$-hierarchy hardness and fixed-parameter tractability; e.g., see~\citet{CFK+15}.} problem, when parameterized by the committee size $k$ (Section~\ref{sec:complexity}). Our proof uses a quite involved reduction from {\sc IndependentSet}, which, on input a graph, defines a profile consisting of several parts and a committee. Some of the parts of the profile guarantee that the only ABCS rule that can make the committee winning has a very particular form: it takes into account only votes with two candidates (ignoring the rest), and mimics the approval-based CC rule (henceforth, simply, the CC rule), a famous rule that is inspired by the work of~\citet{CC83}. Then, the main part of the profile guarantees that the committee is indeed winning under this rule if and only if the graph does not have a large independent set. Our reduction can be modified to give $\mathsf{coW}[1]$-hardness for the following {\em winner verification} problem: given a profile and a committee, is the committee winning under the CC rule? This result strengthens a recent one by~\citet{SDM20}.

We also consider sequential Thiele rules (Section~\ref{sec:seq-Thiele}). These can be thought of as greedy approximations of a subclass of ABCS rules which originate from the work of~\citet{T95}. However, their definition is considerably different from ABCS rules, so that our sample complexity analysis techniques need revision. Still, we are able to show polynomial sample complexity bounds for learning sequential Thiele rules. Interestingly, the problem of deciding whether there is some sequential Thiele rule that makes a given committee winning in a given profile is now fixed-parameter tractable (parameterized by the committee size). Despite this seemingly positive result, we provide evidence that efficient learning is out of reach for sequential Thiele rules as well, by showing NP-hardness. We do so by a novel reduction from a structured version of \textsc{3SAT}, which equates the ordering in which several candidates are greedily included in the winning committee with a boolean assignment to the \textsc{3SAT} variables. As a corollary, our reduction can be modified to yield the first NP-hardness result for the winner verification problem for the sequential CC rule.
\smallskip

\subsection{Related work}
The paper by~\citet{PZPR09} is the most related to ours. Among other results, they prove that the class of single-winner positional scoring rules is efficiently PAC-learnable. We remark that our setting is much more demanding. In particular, the number of possible outcomes is doubly exponential in our case, i.e., $2^{m\choose k}-1$, the number of all possible non-empty sets of winning committees, while it is just $m$ in theirs (where fixed tie-breaking is used to produce a single winning candidate). Hence, even though we have not been able to prove efficiency of learning, the low sample complexity of ABCS rules is rather surprising.

PAC learning in voting has been considered, among other economic paradigms, by~\citet{JZ20} and, in relation to the notion of the distortion, by~\citet{BCH+15}. Actually, the use of sign patterns  has been inspired by the latter paper, even though the particular way in which we employ the result of~\citet{A96} here is different.

More distantly related to our setting, the data-driven approach in the design of voting rules has been followed by a series of papers which focus on particular applications like rating~\citep{CCKV19}, evaluation of online surveys~\citep{BH19}, and peer grading \citep{CKV15,CKV20}. Other foundational work in this direction includes the papers by~\citet{FST18} and~\citet{X13}.

The computational complexity of multiwinner voting rules has received much attention; see the book by~\citet{LS23}. The CC rule has been central in most related studies regarding ABCS rules. \citet{PRZ08} proved that the problem of deciding whether there is a committee that exceeds a given threshold under the CC rule on a given profile is $\mathsf{NP}$-hard. The problem was later proved to be $\mathsf{W}[2]$-hard by~\citet{BSU13}. \citet{SDM20} considered the question of whether a given candidate belongs to a winning committee for a given profile. They also prove that winner verification for the CC rule is $\mathsf{coNP}$-hard, using a reduction from a variant of $3$-{\sc HittingSet}.  We are not aware of published hardness results (of a similar spirit) for sequential Thiele rules.

\subsection{Roadmap}
The rest of the paper is structured as follows. We begin with preliminary definitions and notation in Section~\ref{sec:prelim} and present briefly the necessary background on PAC learning in Section~\ref{sec:pac-background}.  Sections~\ref{sec:learnability}-\ref{sec:seq-Thiele} contain our technical contributions, as outlined in Section~\ref{sec:our-contribution} above. We conclude in Section~\ref{sec:conclusion}, where we highlight our two byproduct results on winner verification; the formal statements of these results and their proofs, which are simplifications of our main hardness proofs, appear in Appendix~\ref{sec:app}.

\section{Preliminaries}\label{sec:prelim}
We consider approval-based voting with a set $\mathcal{N}$ of $n$ {\em voters} (or {\em agents}), each approving a subset from a set $\Sigma$ of $m$ {\em candidates} (or {\em alternatives}). An approval-based multiwinner voting rule is defined for an integer $k$ with $1 < k < m$. It takes as input a profile $P = \{\sigma_i\}_{i \in \mathcal{N}}$, where $\sigma_i\subset \Sigma$ is the non-empty set of alternatives approved by agent $i\in \mathcal{N}$ (or, her {\em approval vote}), and returns one or more $k$-sized subsets of $\Sigma$. We use the term {\em committee} to refer to any $k$-sized set of alternatives; then, the outcome of a multiwinner voting rule is one or more {\em winning} committees.  We are interested in a specific class of multiwinner voting rules called {\em approval-based committee scoring} (ABCS) rules, defined by~\citet{LS21}. These rules are specified by a set of scoring parameters. Using these parameters, an agent's approval vote gives a score to each committee and the winning committees are those that receive the highest total score from all agents.

More formally, an ABCS rule is specified by a bivariate scoring function $f$. The parameter $f(x,y)$ denotes the non-negative score that an approval vote $\sigma$ gives to the committee $C$ when $\sigma$ consists of $y$ alternatives and has $x$ alternatives in common with $C$. Notice that, under this interpretation, the function $f$ needs only be defined over the set of pairs
\begin{align*}
\mathcal{X}_{m,k} &= \left\{(x,y): y\in [m-1], x\in \max\{0,y-m+k\}, ..., \min\{k,y\}\right\}.
\end{align*}
Indeed, an approval vote with $y$ alternatives can intersect with a committee in at least $\max\{0,y-m+k\}$ and at most $\min\{k,y\}$ alternatives. 

Hence, formally $f:\mathcal{X}_{m,k} \rightarrow \R_{\geq 0}$. By definition, $f$ is monotone non-decreasing in its first argument.  To keep the presentation concise, we slightly overload notation and use $f$ to refer both to the scoring function $f$ and the ABCS rule specified by $f$. On input a profile $P=\{\sigma_i\}_{i\in \mathcal{N}}$, the ABCS rule $f$ assigns a score of 
\begin{align*}
\score_f(C,P) &= \sum_{i \in \mathcal{N}} f\left(|C \cap \sigma_i|, |\sigma_i|\right)
\end{align*}
to each committee $C$; then, any committee of maximum score is winning in profile $P$ under rule $f$. We write $f(P)$ for the set of all winning committees in profile $P$. We denote by $\mathcal{F}^{m,k}$ the class of ABCS rules with $m$ alternatives and committee size $k$. We use the term {\em trivial} to refer to the ABCS rule $f$ with $f(x,y)=0$ for every $(x,y)\in \mathcal{X}_{m,k}$; obviously, all committees are winning in any profile under this rule.\footnote{We remark that scoring functions with different parameters may correspond to the very same ABCS rule. Indeed, as \citet{LS21} observe, the ABCS rule $g$ specified by $g(x,y)=c\cdot f(x,y)+d(y)$ for $c>0$ and $d:[m-1]\rightarrow \R_{\geq 0}$ is identical to the ABCS rule $f$, in the sense that, for every profile, they define the same set of winning committees. Throughout the paper, we consider ABCS rules whose scoring function is normalized to have $f(\max\{0,y-m+k\},y)=0$ for every $y\in [m-1]$.  So, the trivial ABCS rule is representative of the rules in which $f(x,y)$ depends only on $y$; all these rules have all committees as winning.}

An important subclass of ABCS rules is that of {\em Thiele} rules. Thiele rules use scoring functions $f$ where the scoring parameter $f(x,y)$ does not depend on $y$. In this case, we can assume that $f$ is univariate, defined over $\{0, 1, ..., k\}$, non-negative, and monotone non-decreasing. A specific Thiele rule that we use extensively is the CC rule that uses $f(0)=0$ and $f(x)=1$ for $x>0$. 

To bypass the necessity of computing the scores of all committees, {\em sequential Thiele rules} have been introduced to approximate ABCS rules by computing a winning committee in a greedy manner. Starting from an empty subcommittee, such rules build a winning committee gradually in $k$ steps; in each step, they include an alternative that increases the score of the current subcommittee the most.  The sequential Thiele rule that uses the univariate scoring function $f$ computes the intermediate score of a set of alternatives $A$ of size {\em up to $k$} on profile $P=\{\sigma_i\}_{i\in \mathcal{N}}$ as  $\score_f(A,P)=\sum_{i\in \mathcal{N}}{f(|A\cap \sigma_i|)}$. Then, given a profile $P$, a committee $C$ is winning under the sequential Thiele rule $f$ in profile $P$ if there is an ordering of the alternatives in $C$, e.g., as $C=\{c_1, c_2, ..., c_k\}$, so that
\begin{align*}
c_i &\in \argmax_{c\in \Sigma\setminus \{c_1, ..., c_{i-1}\}}{\score_f(\{c_1, ..., c_{i-1}\}\cup \{c\},P)},
\end{align*}
for every $i\in [k]$. We denote by $\mathcal{F}^k_{\text{seq}}$ the class of sequential Thiele rules for committee size $k$ and any number of alternatives $m$ higher than $k$. Again, the term {\em trivial} is reserved for the sequential Thiele rule that uses a scoring function $f$ with $f(x)=0$ for every $x$.

We conclude this section by defining the two decision problems we study: \ABCS\ and \SeqThiele. In both, we are given a profile of approval votes $P=\{\sigma_i\}_{i\in \mathcal{N}}$ over the set $\Sigma$ of $m$ alternatives and a $k$-sized subset $C$ of $\Sigma$. Our goal is to decide whether there exists a non-trivial rule $f$ from $\mathcal{F}^{m,k}$ (for \ABCS) or $\mathcal{F}^k_{\text{seq}}$ (for \SeqThiele), so that $C$ is a winning committee in profile $P$ according to $f$. Separate definitions of these problems are given as Definitions~\ref{defn:abcs} and~\ref{defn:seq_thiele} in the sections discussing the complexity of the respective problems (Sections~\ref{sec:complexity} and~\ref{sec:seq-Thiele}).

\section{PAC Learning Background}\label{sec:pac-background}
We follow a standard PAC learning model. In this model, a learning algorithm has to learn a target function from a hypothesis class $\mathcal{H}$ of functions which assign labels from the set $Y$ to the points of a set $Z$. The learning algorithm is given a {\em training set of examples} $T$ consisting of points from the sample space $Z$ paired with the labels from $Y$ that the target function assigns to the respective data point. The data points in the training set are sampled i.i.d.~according to some probability distribution $D$ over $Z$. We consider the {\em realizable case} and assume that there exists a function $h^*\in \mathcal{H}$ that is used to label the examples in the training set as $\{(z,h^*(z))\}_{z\in T}$. The learning algorithm outputs a function $h\in\mathcal{H}$. The error of function $h$ is defined as 
\begin{align*}
\err(h) &= \Pr_{z\sim D}[h(z)\not=h^*(z)].
\end{align*} 
Clearly, $\err(h^*)=0$. The terms ``probably'' and ``approximately correct'' refer to the existence of two parameters $\delta,\epsilon \in (0,1)$, indicating the required {\em confidence} and {\em accuracy} of learning, respectively.

\begin{defn}[PAC learnability]
A hypothesis class $\mathcal{H}$ of functions from set $Z$ to set $Y$ is PAC-learnable if there exist a function $s:(0,1)^2 \rightarrow \N$ ---the sample complexity of $\mathcal{H}$--- and a learning algorithm $\mathcal{A}$ with the following property: For every $\delta,\epsilon\in (0,1)$, every distribution $D$ over $Z$, and every function $h^*$ from $\mathcal{H}$, on input a training set of at least $s(\delta,\epsilon)$ examples generated by $D$ and labelled by $h^*$, the probability (over the choice of the training examples) that algorithm $\mathcal{A}$ returns a hypothesis $h$ of error more than $\epsilon$ is at most $\delta$.
\end{defn}

Extending the relation of the well-known VC dimension with the PAC learnability of boolean functions, \citet{N89} relates the sample complexity of a hypothesis class $\mathcal{H}$ to the notions of graph dimension and Natarajan (or generalized) dimension, both capturing the combinatorial richness of $\mathcal{H}$. To define them, we need to define the notion of {\em shattering} first.

\begin{defn}[shattering]\label{defn:shattering}
Let $\mathcal{H}$ be a class of functions from $Z$ to $Y$ and let $T \subseteq Z$. We say that $\mathcal{H}$ {\em G-shatters} $T$ if there exists a function $g \in \mathcal{H}$ such that
\begin{itemize}
    \item For all $S \subseteq T$, there exists $h_S \in \mathcal{H}$ such that $h_S(z) = g(z)$ for all $z \in S$, and $h_S(z) \neq g(z)$, for all $z \in T \setminus S$.
\end{itemize}
We say that $\mathcal{H}$ {\em N-shatters} $T$ if there exist two functions $g_1,g_2 \in \mathcal{H}$ such that
\begin{enumerate}
	\item For all $z \in T$, $g_1(z) \neq g_2(z)$.
	\item For all $S \subseteq T$, there exists $h_S \in \mathcal{H}$ such that $h_S(z) = g_1(z)$ for all $z \in S$, and $h_S(z) = g_2(z)$, for all $z \in T \setminus S$.
\end{enumerate}
\end{defn}

\begin{defn}[graph and Natarajan dimension]\label{defn:graph_natarajan_dimension}
Let $\mathcal{H}$ be a class of functions from a set $Z$ to a set $Y$. The {\em graph dimension} of $\mathcal{H}$, denoted by $D_G(\mathcal{H})$, is the greatest integer $d$ such that there exists a set of cardinality $d$ that is G-shattered by $\mathcal{H}$. The {\em Natarajan dimension} of $\mathcal{H}$, denoted by $D_N(\mathcal{H})$, is the greatest integer $d$ such that there exists a set of cardinality $d$ that is N-shattered by $\mathcal{H}$.
\end{defn}

We will use the relation of the graph and Natarajan dimension to the sample complexity that is given by the next statement (Theorem~\ref{thm:sample-vs-graph_Natarajan}).  This is a variant of the multiclass fundamental theorem which, in its standard form (e.g., see \citet{SSBD14}), uses only the Natarajan dimension to both upper- and lower-bound the sample complexity.  In the variant we use here, the graph dimension is used instead to upper-bound the sample complexity.\footnote{In the conference version of the paper~\citep{CF22}, we used a formulation of the multiclass fundamental theorem which provides sample complexity upper bounds in terms of the Natarajan dimension and the number of different labels a function from the hypothesis class $\mathcal{H}$ may realize. We prefer to use Theorem~\ref{thm:sample-vs-graph_Natarajan} instead in this revised version, which leads to considerably simpler proofs and better sample complexity bounds in all cases besides the extreme ones where the accuracy parameter $\epsilon$ is exponentially small in terms of $|\mathcal{X}_{m,k}|$ for ABCS rules or in terms of $k$ for sequential Thiele rules.}

\begin{thm}[multiclass fundamental theorem, \citet{DSBDSS11}]\label{thm:sample-vs-graph_Natarajan}
There exist constants $C_1, C_2>0$ such that the hypothesis class $\mathcal{H}$ is PAC-learnable (assuming realizability) with sample complexity $s(\delta,\epsilon)$ that satisfies
\begin{align*}
C_1\cdot \frac{D_N(\mathcal{H})+\ln{(1/\delta)}}{\epsilon}
\leq s(\delta,\epsilon)
\leq C_2\cdot \frac{D_G(\mathcal{H})\cdot \ln(1/\epsilon)+\ln(1/\delta)}{\epsilon}\,.
\end{align*}
\end{thm}

\section{The Learnability of ABCS Rules}\label{sec:learnability}
We are ready to prove that the class $\mathcal{F}^{m,k}$ of ABCS rules is PAC-learnable with sample complexity that depends polynomially on the number of alternatives $m$ and the committee size $k$. 

\begin{thm}\label{thm:H-PAC-learnable}
The class $\mathcal{F}^{m,k}$ of ABCS rules with $m$ alternatives and committee size $k$ is PAC-learnable with sample complexity $s(\delta,\epsilon)$ such that
$$s(\delta,\epsilon) \in \Omega\left(\epsilon^{-1}\left(|\mathcal{X}_{m,k}|+\ln{\left(1/\delta\right)}\right)\right),$$
and 
$$s(\delta,\epsilon) \in \mathcal{O}\left(\epsilon^{-1}\left(|\mathcal{X}_{m,k}|\cdot k\cdot \ln{m}\cdot \ln{\left(1/\epsilon\right)}+\ln{\left(1/\delta\right)}\right)\right).$$
\end{thm}
Notice that $|\mathcal{X}_{m,k}|\in \Theta(k(m-k))$; so the sample complexity grows only polynomially in $m$, $k$, and $1/\epsilon$, and logarithmically in $1/\delta$. Our proof of Theorem~\ref{thm:H-PAC-learnable} will follow by Theorem~\ref{thm:sample-vs-graph_Natarajan} after proving an upper bound of $\mathcal{O}\left(|\mathcal{X}_{m,k}|\cdot k\cdot \ln{m}\right)$ on the graph dimension (Lemma~\ref{lem:non_sequential_upper_bound_graph_dim}) and a lower bound of $\Omega\left(|\mathcal{X}_{m,k}|\right)$ on the Natarajan dimension of $\mathcal{F}^{m,k}$ (Lemma~\ref{lem:non_sequential_lower_bound_generalized_dim}).

\subsection{Upper-bounding the graph dimension}
To bound the graph dimension, we will use an important result in algebraic combinatorics that bounds the number of different sign patterns a set of polynomials may have. Consider a set $\mathcal{L}$ of $K$ polynomials $p_1, p_2, ..., p_K$, each defined over the $\ell$ real variables $x_1, x_2, ..., x_{\ell}$ (i.e., $p_i:\R^{\ell}\rightarrow \R$ for $i\in [K]$). A sign pattern $\mathbf{s}$ is just a vector of values in $\{-1,0,+1\}$ with $K$ entries. We say that the set of polynomials $\mathcal{L}$ realizes the sign pattern $\mathbf{s}$ if there exist values $x^*_1, x^*_2, ..., x^*_\ell$ for the variables $x_1, x_2, ..., x_\ell$ such that $\text{sgn}(p_i(x^*_1, x^*_2, ..., x^*_\ell))=\mathbf{s}_i$, for $i=1, 2, ..., K$. Here, $\text{sgn}$ is the signum function returning $-1$, $0$, or $+1$, depending on whether its argument is negative, zero, or positive.

Clearly, the number of different sign patterns $K$ polynomials may realize is at most $3^K$. Usually, this is a very weak upper bound; \citet{A96} provides a much better bound, extending a previous statement due to \citet{W68}.

\begin{thm}[\citet{A96}, \citet{W68}]\label{thm:alon_extension_warren}
The number of different sign patterns a set of $K$ polynomials of degree $\tau$ over $\ell$ real variables may realize is at most $\left(\frac{8e\tau K}{\ell}\right)^\ell$.
\end{thm}

Using Theorem~\ref{thm:alon_extension_warren}, we can prove the next upper bound on the graph dimension $D_G(\mathcal{F}^{m,k})$.

\begin{lemma}\label{lem:non_sequential_upper_bound_graph_dim}
$D_G(\mathcal{F}^{m,k}) \in \mathcal{O}(|\mathcal{X}_{m,k}|k \log{m})$.
\end{lemma}

\begin{proof}
Assume that the graph dimension of $\mathcal{F}^{m,k}$ is $N$ with $N\geq 4\cdot |\mathcal{X}_{m,k}|$ since, clearly, the upper bound on $D_G(\mathcal{F}^{m,k})$ holds if $N<4\cdot |\mathcal{X}_{m,k}|$.  According to Definition~\ref{defn:graph_natarajan_dimension}, we thus have a set of $N$ different profiles $\{P_j\}_{j \in [N]}$ and a voting rule $g \in \mathcal{F}^{m,k}$ such that for any subset $S \subseteq [N]$ there exists a rule $h_S \in \mathcal{F}^{m,k}$ with the property that
\begin{itemize}
\item $h_S(P_j) = g(P_j)$ for all $j \in S$, and
\item $h_S(P_j) \neq g(P_j)$ for all $j \in [N]\setminus S$.
\end{itemize}
Let $\mathcal{C}^k$ be the set of all $k$-sized committees. For every pair of committees $C,D \in \mathcal{C}^k$ and every $j \in [N]$, we define
$$L^j_{C, D}(s) = \score_s(C, P_j) - \score_s(D, P_j),$$
where $s$ is any scoring function specifying a voting rule in $\mathcal{F}^{m,k}$. Let $P_j=\{\sigma^j_i\}_{i\in \mathcal{N}}$; then 
\begin{align*}
L^j_{C, D}(s)
= \sum_{i\in \mathcal{N}}s\left(|C\cap \sigma^j_i|, |\sigma^j_i|\right)
- \sum_{i\in \mathcal{N}}s\left(|D\cap \sigma^j_i|, |\sigma^j_i|\right).
\end{align*}
Hence, $L^j_{C,D}(s)$ is a linear function (a polynomial of degree $1$) on the variables $s(x,y)$ for $(x,y)\in \mathcal{X}_{m,k}$. Let $\mathcal{L} = \{L^j_{C,D}(s): j \in [N] \mbox{ and } C,D \in \mathcal{C}^k\}$ be the set of linear functions defined for the $N$ different profiles and all pairs of $k$-sized committees. We note that $|\mathcal{L}| \leq N m^{2k}$ since $|\mathcal{C}^k| = \binom{m}{k} \leq m^k$.

Now, consider two different subsets $S$ and $S'$ of $[N]$ such that $S\not\subseteq S'$ (notice that this is without loss of generality). Let $j^*$ be such that $j^*\in S\setminus S'$. Since $\mathcal{F}^{m,k}$ G-shatters the set of profiles $\{P_j\}_{j \in [N]}$ by our assumption, there exist two rules $h_S, h_{S'} \in \mathcal{F}^{m,k}$ such that $h_S(P_{j^*}) = g(P_{j^*})$ and $h_{S'}(P_{j^*}) \neq g(P_{j^*})$.  Hence,  there must be a pair of committees $C,D \in \mathcal{C}^k$ such that $C \in h_{S}(P_{j^*}), D \in h_{S'}(P_{j^*})$ and either $C \notin h_{S'}(P_{j^*})$ or $D \notin h_{S}(P_{j^*})$ (not necessarily exclusively). Then,
\begin{align*}
\text{sgn}\left(L^{j^*}_{C, D}(h_{S})\right)
= \text{sgn}\left(\score_{h_{S}}(C, P_{j^*}) - \score_{h_{S}}(D, P_{j^*})\right)
= \begin{cases}
  	0, & \text{if } D \in h_{S}(P_{j^*})\\
  	+ 1, & \text{if } D \notin h_{S}(P_{j^*})
  \end{cases}
\end{align*}
and\
\begin{align*}
\text{sgn}\left(L^{j^*}_{C,D}(h_{S'})\right)
= \text{sgn}\left(\score_{h_{S'}}(C,P_{j^*}) - \score_{h_{S'}}(D,P_{j^*})\right)
= \begin{cases}
  	0, & \text{if } C \in h_{S'}(P_{j^*})\\
  	- 1, & \text{if } C \notin h_{S'}(P_{j^*})
  \end{cases}
\end{align*}
Since it is either $C \notin h_{S'}(P_{j^*})$ or $D \notin h_{S}(P_{j^*})$, we get that $\text{sgn}\left(L^{j^*}_{C,D}(h_{S})\right)\not=\text{sgn}\left(L^{j^*}_{C,D}(h_{S'})\right)$. Hence, each of the $2^N$ voting rules $h_S$ for $S\subseteq [N]$ ---corresponding to a distinct assignment of values $s(x,y)$ for $(x,y)\in\mathcal{X}_{m,k}$--- yields a different sign pattern to the set of polynomials $\mathcal{L}$. We now apply Theorem~\ref{thm:alon_extension_warren} to $\mathcal{L}$ for $K = N m^{2k}, \tau = 1, \ell = |\mathcal{X}_{m,k}|$. This gives an upper bound of
$\left(\frac{8e N m^{2k}}{|\mathcal{X}_{m,k}|}\right)^{|\mathcal{X}_{m,k}|}$
on the number of different sign patterns with entries in $\{-1,0,+1\}$ for the set of polynomials $\mathcal{L}$. Hence, 
\begin{align*}
2^{N} &\leq \left(\frac{8e N m^{2k}}{|\mathcal{X}_{m,k}|}\right)^{|\mathcal{X}_{m,k}|}
\end{align*}
and, equivalently,
\begin{align*}
\frac{|\mathcal{X}_{m,k}|}{N} \cdot 2^{N/|\mathcal{X}_{m,k}|}\leq 8e m^{2k}.
\end{align*}
We now apply the property $2^{z/2}\geq z$ for $z\geq 4$.  Using $z=\frac{N}{|\mathcal{X}_{m,k}|}$, combined with our assumption that $N\geq 4|\mathcal{X}_{m,k}|$, we get
\begin{align*}
2^{\frac{1}{2}\cdot N/|\mathcal{X}_{m,k}|}\leq 8e m^{2k}
\end{align*}
and, $N\leq 2 |\mathcal{X}_{m,k}| \log{\left(8e m^{2k}\right)}$, as desired.
\end{proof}

\subsection{Lower bounding the Natarajan dimension}
We now prove a lower bound on $D_N(\mathcal{F}^{m,k})$. In our proof, we construct a large set of profiles that can be $N$-shattered by the set of ABCS rules $\mathcal{F}^{m,k}$. 

\begin{lemma}\label{lem:non_sequential_lower_bound_generalized_dim}
$D_N(\mathcal{F}^{m,k}) \in \Omega(|\mathcal{X}_{m,k}|)$.
\end{lemma}

\begin{proof}
For a given $m\geq 3$ and $k$ such that $2\leq k \leq m-1$, consider the set of alternatives
$$\Sigma = \{a, b_1, \ldots, b_{k-1},c,d_1,\ldots, d_{m-k-1}\}.$$
Our goal is to define a set of profiles, where for each profile we are able to pick rules from $\mathcal{F}^{m,k}$ such that either committee $A = \{a, b_1, \ldots, b_{k-1}\}$ or committee $C = \{b_1, \ldots, b_{k-1}, c\}$ is the single winning committee under the respective rule.

Let $\mathcal{T}_{m,k}$ be the following set of pairs:
\begin{align*}
\mathcal{T}_{m,k} &= \Big\{(x,y) : y \in \{2,\ldots, m-1\},
x \in \big\{1+\max\{0, y-m+k\}, \ldots, \min\{k,y\}\big\}\Big\}
\setminus \{(k,k)\}.
\end{align*}
Even though some of the pairs of set $\mathcal{X}_{m,k}$ have been omitted from $\mathcal{T}_{m,k}$, they have asymptotically the same size as the next lemma indicates.

\begin{lemma}\label{lem:T-vs-X}
$|\mathcal{T}_{m,k}| \in \Omega(|\mathcal{X}_{m,k}|)$.
\end{lemma}

\begin{proof}
The lemma clearly holds for $m<4$ since both $\mathcal{X}_{m,k}$ and $\mathcal{T}_{m,k}$ have constant size in this case. In the following, we assume $m\geq 4$. Observe that the $(x,y)$ pairs of $\mathcal{X}_{m,k}$ that are missing from $\mathcal{T}_{m,k}$ are the following: $(\max\{0,y-m+k\},y)$ for $y=2, ..., m-1$, $(0,1)$, $(1,1)$, and $(k,k)$. Hence, 
\begin{align}\label{eq:T-vs-X}
|\mathcal{T}_{m,k}| &= |\mathcal{X}_{m,k}|-m-1
\end{align}
Also, observe that 
\begin{align}\label{eq:X_mk-vs-X_m1}
|\mathcal{X}_{m,k}|\geq |\mathcal{X}_{m,2}| &=3m-5\geq \frac{7}{5}(m+1).
\end{align}
The second inequality follows since $m\geq 4$. Now, using equations (\ref{eq:T-vs-X}) and (\ref{eq:X_mk-vs-X_m1}), we get
\begin{align*}
|\mathcal{T}_{m,k}| &\geq |\mathcal{X}_{m,k}|-\frac{5}{7}|\mathcal{X}_{m,k}|=\frac{2}{7}|\mathcal{X}_{m,k}|,
\end{align*}
as desired.
\end{proof}

We now define the set of profiles $\{P_{xy}\}_{(x,y) \in \mathcal{T}_{m,k}}$. Each profile $P_{xy}$ contains four approval votes:
$$\sigma^{xy}_1 = \{a\}, \sigma^{xy}_2 = A, \sigma^{xy}_3 = C,$$
and
$$\sigma^{xy}_4 = \{b_1, \ldots, b_{x-1}, c, d_1, \ldots, d_{y-x}\}.$$
We introduce the family of rules $F \subseteq \mathcal{F}^{m,k}$ which, for every subset $S\subseteq \mathcal{T}_{m,k}$, contains the voting rule $h_S$ defined as:
\begin{align*}
h_S(1,1) &= 1,\\
h_S(k,k) &= 4k-1,\\
h_S(\max\{0,y-m+k\},y) &= 0, \quad\,\,\, \text{for } y \in [m-1],\\
h_S(x,y) - h_S(x-1,y) &= 
\begin{cases}
0, & (x,y)\in S,\\
2, & (x,y)\in \mathcal{T}_{m,k} \setminus S.
\end{cases} 
\end{align*}
Notice that the function $h_S$ is monotonically non-decreasing in its first argument, as required by the definition of voting rules in $\mathcal{F}^{m,k}$.

We will show that the family $F$ $N$-shatters the set of profiles $\{P_{xy}\}_{(x,y) \in \mathcal{T}_{m,k}}$. To do so, we will make use of the following two lemmas. Lemma~\ref{lem:non_sequential_lower_bound_A_beats_X} guarantees that no other committee besides $A$ and $C$ is ever winning in any profile of $\{P_{xy}\}_{(x,y) \in \mathcal{T}_{m,k}}$. Lemma~\ref{lem:non_sequential_lower_bound_sc_diff_A_C} identifies the winning committee among $A$ and $C$ in each of these profiles for every voting rule in set $F$.
 
\begin{lemma}\label{lem:non_sequential_lower_bound_A_beats_X}
For every committee $X \neq A,C$, every profile $P_{xy} \in \{P_{xy}\}_{(x,y) \in \mathcal{T}_{m,k}}$ and any voting rule $h \in F$, it holds that $\score_h(A, P_{xy}) - \score_h(X, P_{xy}) \geq 1$.
\end{lemma}

\begin{proof}
Let us first assume that at least one of the following two assumptions hold: either $y\not=k$ or $|X\cap \sigma_4^{xy}|<k$. Notice that $h(|X\cap \sigma_4^{xy}|,y)\leq 2k$ in this case, due to the definition of the family $F$. Then, since $X$ is different than both $A$ and $C$, we have that
\begin{align*}
\score_h(X,P_{xy}) &= h(|X\cap\{a\}|,1)+h(|X\cap A|,k)+h(|X\cap C|,k)+h(|X\cap \sigma_4^{xy}|,y)\\
&\leq h(1,1)+2 h(k-1,k)+2k,
\end{align*}
while
\begin{align*}
\score_h(A,P_{xy})=h(1,1)+h(k,k)+h(k-1,k)+h(x-1,y).
\end{align*}
Hence, 
\begin{align*}
\score_h(A,P_{xy}) - \score_h(X,P_{xy}) &\geq h(k,k)+h(x-1,y)-h(k-1,k)-2k \geq 1
\end{align*}
since $h(x-1,y)\geq 0$, $h(k,k)=4k-1$, and $h(k-1,k) \leq 2k-2$.

Now, assume that $y=k$ and $|X\cap \sigma_4^{xy}|=k$; then $X=\sigma_4^{xy}$. We have
\begin{align*}
\score_h(X,P_{xy}) &= h(0,1)+h(x-1,k)+h(x,k)+h(k,k)\\&\leq h(0,1)+h(x-1,k)+h(k-1,k)+h(k,k).
\end{align*}
The inequality follows since $(k,k)\not\in \mathcal{T}_{m,k}$ and, thus, $x\leq k-1$, and by the monotonicity of voting rule $h$. Hence, 
\begin{align*}
\score_h(A,P_{xy}) - \score_h(X,P_{xy}) &\geq h(1,1)-h(0,1) = 1. \qedhere
\end{align*}
\end{proof}

\begin{lemma}\label{lem:non_sequential_lower_bound_sc_diff_A_C}
Let $S\subseteq \mathcal{T}_{m,k}$. For every profile $P_{xy} \in \{P_{xy}\}_{(x,y) \in \mathcal{T}_{m,k}}$, it holds that
\begin{align*}
\score_{h_S}(A, P_{xy}) - \score_{h_S}(C, P_{xy}) &=
\begin{cases}
1, & (x,y) \in S\\
-1, & (x,y) \in \mathcal{T}_{m,k}\setminus S
\end{cases}
\end{align*}
\end{lemma}

\begin{proof}
By the definition of profile $P_{xy}$, we have
\begin{align*}
\score_{h_S}(A,P_{xy}) &= h_S(1,1) + h_S(k,k) + h_S(k-1,k) + h_S(x-1,y),
\end{align*}
and
\begin{align*}
\score_{h_S}(C, P_{xy})
&= h_S(0,1)+h_S(k-1,k) + h_S(k,k) + h_S(x,y).
\end{align*}
Recall that $h_S(0,1)=0$ and $h_S(1,1)=1$. Thereby, 
\begin{align*}
\score_{h_S}(A,P_{xy}) - \score_{h_S}(C, P_{xy})
&= 1 - \left(h_S(x,y) - h_S(x-1,y)\right),
\end{align*}
and the lemma follows from the definition of voting rule $h_S$.
\end{proof}

Together, Lemmas~\ref{lem:non_sequential_lower_bound_A_beats_X} and~\ref{lem:non_sequential_lower_bound_sc_diff_A_C} imply that, when applied to profile $P_{xy}$, the voting rule $h_S$ returns
\begin{itemize}
\item $A$ as the unique winning committee if $(x,y)\in S$, and
\item $C$ as the unique winning committee if $(x,y)\in \mathcal{T}_{m,k}\setminus S$.
\end{itemize}
By Definition~\ref{defn:shattering}, this implies that the family $F$ (and, consequently, the family $\mathcal{F}^{m,k}$) $N$-shatters the set of profiles $\{P_{xy}\}_{(x,y)\in \mathcal{T}_{m,k}}$. Indeed, it suffices to define functions $g_1$ and $g_2$ as $g_1=h_{\mathcal{T}_{m,k}}$ and $g_2=h_{\emptyset}$, while the set of profiles $\{P_{xy}\}_{(x,y)\in \mathcal{T}_{m,k}}$ plays the role of set $T$ in Definition~\ref{defn:shattering}. By Definition~\ref{defn:graph_natarajan_dimension}, we conclude that the Natarajan dimension of $\mathcal{F}^{m,k}$ is at least $|\mathcal{T}_{m,k}|$. Lemma~\ref{lem:non_sequential_lower_bound_generalized_dim} now follows from Lemma~\ref{lem:T-vs-X}.
\end{proof}

\section{The Complexity of \ABCS}\label{sec:complexity}
Unfortunately, despite the low sample complexity of the class of ABCS rules, learning from samples is notoriously hard. We prove this for \ABCS, which captures the elementary task of learning from a single sample. We repeat the formal definition of \ABCS\ here.
\begin{defn}[\ABCS]\label{defn:abcs}
Given a profile of approval votes $P=\{\sigma_i\}_{i\in \mathcal{N}}$ over the set $\Sigma$ of $m$ alternatives and a $k$-sized subset $C$ of $\Sigma$, decide whether there exists a non-trivial rule $f$ from $\mathcal{F}^{m,k}$, so that $C$ is a winning committee in profile $P$ according to $f$.
\end{defn}
The next statement uses a polynomial-time reduction from (the complement of) the \textsc{IndependentSet} problem.
\begin{defn}[\textsc{IndependentSet}]\label{defn:independent_set}
Given a graph $G$ and a positive integer $K$, decide whether $G$ contains a set of at least $K$ nodes that form an independent set.
\end{defn}
\textsc{IndependentSet} is known to be $\mathsf{W}[1]$-hard, parameterized by the independent set size~\citep[Theorem 13.18]{CFK+15}.

\begin{thm}\label{thm:ABCS_coW1_hard}
\ABCS\ parameterized by the committee size $k$ is $\mathsf{coW}[1]$-hard.
\end{thm}

\begin{proof}
For a given instance of \textsc{IndependentSet} consisting of a graph $G$ and an integer $K$, we construct an instance of \ABCS\ with $k=K$ such that there is a non-trivial rule $f\in \mathcal{F}^{m,k}$ that outputs $A$ as a winning committee in $P$ if and only if $G$ contains no independent set of size $K$.

Let $\Delta$ denote the maximum degree among the vertices of $G$. We can assume that $\Delta\geq 2$, since \textsc{IndependentSet} would be trivially solvable in polynomial time otherwise. As a first step in our construction, we modify $G$ to another graph $G'$ as follows. For every vertex $v\in V$, we add $\Delta-\deg(v)$ dummy vertices that are adjacent only to $v$. Let $G'=(V',E')$ be the resulting graph and let $|V'|=r$. We note that $G'$ may contain more independent sets of size $k$ than $G$. However, observe that each additional $k$-sized independent set in $G'$ contains at least one dummy vertex of degree 1, whereas the vertices of $G$ (including those vertices belonging to any independent set in $G$) have degree $\Delta\geq 2$ in $G'$. As we will see in the later parts of our proof, the construction of our \ABCS\ instance is not affected by the existence of additional $k$-sized independent sets in $G'$ due to this difference in vertex degree.

Without loss of generality, we can assume that $V'$ is the set of positive integers in $[r]$. The set of alternatives $\Sigma$ consists of alternatives $a_i$ and $b_i$ for every vertex $i\in V'$, and the additional alternatives $c$ and $d$. Let $A=\{a_1, a_2, ..., a_k\}$. The profile $P$ consists of three parts: 
\begin{itemize}
\item Part 1 consists of a vote $\{b_i, b_j\}$ for every edge $(i,j)\in E'$. 
\item Part 2 consists of $k\Delta-1$ copies of each of the following votes: vote $\{a_i,b_j\}$ for every $i,j\in [r]$, votes $\{a_i,c\}$ and $\{b_i,d\}$ for every $i\in [r]$, vote $\{a_1, d\}$, and vote $\{c,d\}$. 
\item Part 3 consists of a vote containing alternative $d$, alternatives $a_1$, $a_2$, ..., $a_{x-1}$, and $y-x$ additional alternatives among alternatives $a_{k+1}$, $a_{k+2}$, ..., $a_r$, $b_1$, ..., $b_r$, for every $(x,y)\in \widehat{\mathcal{X}}=\mathcal{X}_{m,k} \setminus \left(\{(\max\{0,y-m+k\},y): y\in [m-1]\} \cup \{(1,2), (2,2)\}\right)$.
\end{itemize}
We use $P_1$, $P_2$, and $P_3$ to denote the three subprofiles of votes in part 1, 2, and 3, respectively.

Parts 2 and 3 of profile $P$ have important properties that are given in Lemmas~\ref{lem:C_almost_beats_A_in_P3} and~\ref{lem:C_almost_beats_A_in_P2}.

\begin{lemma}\label{lem:C_almost_beats_A_in_P3}
Let $f\in \mathcal{F}^{m,k}$ and $C=\{a_1, ..., a_{k-1},d\}$. Then, $\score_f(A,P_3)=\score_f(C,P_3)$ if $f(x,y)=0$ for every $(x,y)\in \mathcal{X}_{m,k}\setminus \{(1,2), (2,2)\}$ and $\score_f(A,P_3)<\score_f(C,P_3)$, otherwise.
\end{lemma}

\begin{proof}
Committees $A$ and $C$ get scores of $f(x-1,y)$ and $f(x,y)$ from the vote of subprofile $P_3$ corresponding to the pair $(x,y)\in \widehat{\mathcal{X}}$. Hence,
\begin{align*}
\score_f(A,P_3)- \score_f(C,P_3) &= \sum_{(x,y)\in \widehat{\mathcal{X}}}{\left(f(x-1,y)-f(x,y)\right)}\\
&=-\sum_{y\in [m-1]\setminus\{2\}}{\sum_{x=1+\max\{0,y-m+k\}}^{\min\{k,y\}}{\left(f(x,y)-f(x-1,y)\right)}}\\
&=-\sum_{y\in [m-1]\setminus\{2\}}{f(\min\{k,y\},y)}.
\end{align*}
Thus, the difference $\score_f(A,P_3)- \score_f(C,P_3)$ is equal to zero if $f(x,y)=0$ for every $(x,y)\in \mathcal{X}_{m,k}\setminus \{(1,2), (2,2)\}$ and negative otherwise. \end{proof}

\begin{lemma}\label{lem:C_almost_beats_A_in_P2}
Let $f\in \mathcal{F}^{m,k}$ and $C=\{a_1, ..., a_{k-1},d\}$. Then, $\score_f(A,P_2)=\score_f(C,P_2)$ if $f(1,2)=f(2,2)$ and $\score_f(A,P_2)<\score_f(C,P_2)$, otherwise.
\end{lemma}

\begin{proof}
In the subprofile $P_2$, committee $A$ gets score $f(1,2)$ from each of the $k\Delta-1$ copies of vote $\{a_i,b_j\}$ for $i\in [k]$ and $j\in [r]$, from each of the $k\Delta-1$ copies of vote $\{a_i,c\}$ for $i\in [k]$ and from each of the $k\Delta-1$ copies of $\{a_1,d\}$, i.e., $(k\Delta-1)(kr+k+1)f(1,2)$ in total. Committee $C$ gets score $f(1,2)$ from each of the $k\Delta-1$ copies of vote $\{a_i,b_j\}$ for $i\in [k-1]$ and $j\in [r]$, from each of the $k\Delta-1$ copies of vote $\{a_i,c\}$ for $i\in [k-1]$, from each of the $k\Delta-1$ copies of $\{b_i,d\}$ for $i\in [r]$, and from each of the $k\Delta-1$ copies of vote $\{c,d\}$, and $f(2,2)$ from each of the $k\Delta-1$ copies of vote $\{a_1,d\}$, i.e., $(k\Delta-1)(kr+k)f(1,2)+(k\Delta-1)f(2,2)$. Hence,
\begin{align*}
\score_f(A,P_2)-\score_f(A,P_2) &= (k\Delta-1)(f(1,2)-f(2,2)),
\end{align*}
which yields the lemma.
\end{proof}

As no vote in part 1 of profile $P$ includes any alternatives in committees $A$ and $C$, Lemmas~\ref{lem:C_almost_beats_A_in_P3} and~\ref{lem:C_almost_beats_A_in_P2} imply that a non-trivial rule $f\in \mathcal{F}^{m,k}$ can make committee $A$ winning in $P$ only if 
it satisfies $f(1,2)=f(2,2)>0$ and $f(x,y)=0$ for any pair $(x,y)$ of $\mathcal{X}_{m,k}$ different than $(1,2)$ and $(2,2)$. We complete the proof assuming ---without loss of generality--- that $f$ furthermore satisfies $f(1,2)=f(2,2)=1$.

\begin{claim}\label{claim:score-of-A}
It holds that $\score_f(A,P)=(k\Delta-1)(kr+k+1)$.
\end{claim}

\begin{proof}
Committee $A$ gets one point from the $k\Delta-1$ copies of vote $\{a_i,b_j\}$ for $i\in [k]$ and $j\in [r]$, the $k\Delta-1$ copies of vote $\{a_i,c\}$ for $i\in [k]$, and the $k\Delta-1$ copies of vote $\{a_1,d\}$.
\end{proof}

Consider a committee $B$ and let $t$ be the number of its alternatives from $\{b_1, ..., b_r\}$, and $\lambda$, $\mu$, and $\nu$ be binary variables indicating whether alternative $c$, $d$, and $a_1$ belongs to $B$, respectively.

\begin{lemma}\label{lem:score-of-B}
If $t<k$, then $\score_f(B,P) \leq (k\Delta-1)(kr+k+1)$.
\end{lemma}

\begin{proof}
We will show that 
\begin{align}\label{eq:score-of-B}
\score_f(B,P) & \leq (k\Delta-1)(\one\{t>0\}+kr+k+\mu+(1-\mu)\nu-(t+\lambda)(k-t-\lambda)).
\end{align}
Inequality (\ref{eq:score-of-B}) yields the claim by distinguishing between three cases:
\begin{itemize}
\item If $t+\lambda=0$, then the terms $\one\{t>0\}$ and $(t+\lambda)(k-t-\lambda)$ are equal to $0$ and $\mu+\nu-\mu\nu\leq 1$.
\item If $k-t-\lambda=0$, then $B$ contains only alternatives from $\{b_1, ..., b_r\}$ and alternative $c$. Hence, $\mu$, $\nu$, and term $(t+\lambda)(k-t-\lambda)$ are all equal to $0$.
\item If $t+\lambda>0$ and $k-t-\lambda>0$, then $(t+\lambda)(k-t-\lambda)\geq 1$, while $\mu+\nu-\mu\nu\leq 1$.
\end{itemize}
Due to (\ref{eq:score-of-B}), all these cases yield the desired inequality for $\score_f(B,P)$.

To prove (\ref{eq:score-of-B}), first observe that $B$ gets at most $\Delta$ points for each of the $t$ alternatives from set $\{b_1, ..., b_r\}$ that $B$ contains. Hence, by the assumption that $t<k$, we obtain that
\begin{align}\label{eq:score-of-B-in-P3}
\score_f(B,P_1) &\leq t\Delta \leq (k\Delta-1)\one\{t>0\}.
\end{align}
To bound $\score_f(B,P_2)$, observe that $B$ gets a point for each of the $k\Delta-1$ copies of
\begin{itemize}
\item vote $\{a_i,b_j\}$ for $i,j\in [r]$ with either $a_i\in B$ or (not exclusively) $b_j\in B$. The total number of these votes is $(k\Delta-1)(tr+(k-t-\lambda-\mu)r-t(k-t-\lambda-\mu))$. 
\item vote $\{a_i,c\}$ for $i\in [r]$ with $c\in B$ or $a_i\in B$. The total number of these votes is $(k\Delta-1)(k-t-\lambda-\mu+r\lambda-\lambda(k-t-\lambda-\mu))$.
\item vote $\{b_i,d\}$ for $i\in [r]$ with $d\in B$ or $b_i\in B$. The total number of these votes is $(k\Delta-1)(r\mu+t-\mu t)$.
\item vote $\{a_1,d\}$ if at least one of $a_1$ and $d$ belongs to $B$. The number of these votes is $(k\Delta-1)(\mu+\nu-\mu\nu)$.
\item vote $\{c,d\}$ if at least one of $c$ and $d$ belongs to $B$. The number of these votes is $(k\Delta-1)(\lambda+\mu-\lambda\mu)$.
\end{itemize}
Summing these numbers of votes, we get
\begin{align}\label{eq:score-of-B-in-P2}
\score_f(B,P_2) &= (k\Delta-1)(kr+k+\mu+\nu-\mu\nu-(t+\lambda)(k-t-\lambda)).
\end{align}
Inequality (\ref{eq:score-of-B}) then follows by summing equations (\ref{eq:score-of-B-in-P3}) and (\ref{eq:score-of-B-in-P2}) and since $\score_f(B,P_3)=0$.
\end{proof}

By Claim~\ref{claim:score-of-A} and Lemma~\ref{lem:score-of-B}, if committee $B$ has score higher than $\score_f(A,P)$, then it must be $t=k$. We conclude the proof by reasoning about $\score_f(B,P)$ in this case.

\begin{claim}\label{claim:score-of-B-in-P2}
Let $B$ be a committee with $t=k$. Then, $\score_f(B,P_2)=(k\Delta-1)(kr+k)$.
\end{claim}

\begin{proof}
In part 2 of the profile, committee $B$ gets one point from the $k\Delta-1$ copies of vote $\{a_i,b_j\}$ for $i\in [r]$ and $b_j\in B$ and the $k\Delta-1$ copies of vote $\{b_i,d\}$ for $b_i\in B$.
\end{proof}

\begin{lemma}\label{lem:score-of-B-in-P1}
Consider any committee $B$ with $t=k$. If $G$ has no independent set of size $k$, then $\score_f(B,P_1)\leq k\Delta-1$.
\end{lemma}

\begin{proof}
Let $S$ be the set of vertices in $G'$ to which the alternatives in $B$ correspond. Then, $\score_f(B,P_1)$ is equal to the number of edges in $G'$ that are incident to the vertices of $S$. These vertices have degree either $1$ or $\Delta$. If one of them has degree $1$, then $\score_f(B,P_1)\leq (k-1)\Delta+1\leq k\Delta-1$. Otherwise, if all of them have degree $\Delta$ in $G'$, then they correspond to vertices of $G$. Since $G$ has no independent set of size $k$, at least two vertices of $S$ are connected by an edge in $G$ and, consequently, in $G'$. Hence, the number of edges incident to the vertices of $S$ and, consequently, $\score_f(B,P_1)$ is at most $k\Delta-1$.
\end{proof}

By Claim~\ref{claim:score-of-B-in-P2} and Lemma~\ref{lem:score-of-B-in-P1}, we obtain that if $G$ has no independent set of size $k$, then $\score_f(B,P)\leq (k\Delta-1)(kr+k+1)$. Thus, by Claim~\ref{claim:score-of-A}, $A$ is a winning committee in this case. 

Now, assume that $G$ has an independent set of size $k$. This implies that $G'$ has an independent set $S$ of $k$ vertices of degree $\Delta$. Now, consider the committee $B$ consisting of the alternatives that correspond to the vertices of $S$. As the number of edges that are incident to vertices of $S$ is $k\Delta$, we have that $\score_f(B,P_1)=k\Delta$ as well. Then, by Claims~\ref{claim:score-of-A} and~\ref{claim:score-of-B-in-P2}, we have $\score_f(B,P)=1+(k\Delta-1)(kr+k+1)>\score_f(A,P)$ indicating that $A$ is not winning. The proof of correctness of our reduction is now complete.
\end{proof}

By restricting the profile to have only part 1 and a simplified variant of part 2, our reduction yields $\mathsf{coW}[1]$-hardness of winner verification for the CC rule as stated in Theorem~\ref{thm:byproduct_result-1} (see Appendix~\ref{sec:app} for the detailed statement and proof).

\section{Sequential Thiele Rules}\label{sec:seq-Thiele}
We now turn our attention to the PAC learnability of sequential Thiele rules and the complexity of the related elementary learning task. We repeat the formal definition of \SeqThiele\ here.
\begin{defn}[\SeqThiele]\label{defn:seq_thiele}
Given a profile of approval votes $P=\{\sigma_i\}_{i\in \mathcal{N}}$ over the set $\Sigma$ of $m$ alternatives and a $k$-sized subset $C$ of $\Sigma$, decide whether there exists a non-trivial rule $f$ from $\mathcal{F}^k_{\text{seq}}$, so that $C$ is a winning committee in profile $P$ according to $f$.
\end{defn}
In Section~\ref{sec:learnability}, we saw how the sign of a single linear function can be used to compare the score of two committees in a profile according to an ABCS rule. Due to the different definition of sequential Thiele rules, such a direct comparison is not possible. Still, deciding whether a committee is winning can be done by examining the signs of a {\em block} of linear functions. This will be our main tool to show that \SeqThiele\ is in FPT and that the class $\mathcal{F}^k_{\text{seq}}$ is PAC-learnable.

Assume an arbitrary ordering of the alternatives in $\Sigma$. For a committee $A$ and integer $i\in [k]$, we denote by $A(i)$ the $i$-th alternative of committee $A$ (according to the assumed ordering). For a committee $A$, permutation $\pi:[k]\rightarrow [k]$, and integer $i\in [k]$, the notation $A[\pi,i]$ is used to denote the set of alternatives $\cup_{j=1}^i{\{A(\pi(j))}\}$.

Now, assume that the sequential Thiele rule $s$ returns committee $A$ as winning when applied on profile $P$. Assume that the order in which rule $s$ decides the alternatives in $A$ as winning is given by permutation $\pi$:~at step $i$, the rule includes alternative $A(\pi(i))$ in the winning committee. By the definition of the sequential Thiele rule $s$, 
this decision can be expressed by the set of inequalities
\begin{align}\label{eq:block}
\score_s(A[\pi,i],P)-\score_s(A[\pi,i-1]\cup\{a\},P)\geq 0,  
\end{align}
for every alternative $a\in \Sigma\setminus A[\pi,i]$. Non-negativity is necessary and sufficient so that alternative $A(\pi(i))$ is (weakly) preferred for inclusion in the winning committee at step $i$ over any alternative $a\in \Sigma\setminus A[\pi,i]$.

For a sequential Thiele rule $s$, committee $A$, and permutation $\pi$, we define the block $B^P_{A,\pi}(s)$ consisting of the LHS expression of equation (\ref{eq:block}) 
for every $i\in [k]$ and every alternative $a\in \Sigma\setminus A[\pi,i]$. By the discussion above, committee $A$ is winning in profile $P$ under rule $s$  if and only if there is a permutation $\pi$ so that all expressions in block $B^P_{A,\pi}(s)$ are non-negative. Otherwise, if the block $B^P_{A,\pi}(s)$ contains a negative expression for every permutation $\pi$, committee $A$ is not winning.

We can use this observation to show that \SeqThiele\ can be solved in time $k!\cdot\text{poly}(m,n)$ and, hence, is fixed-parameter tractable. This can be done as follows. For each of the $k!$ permutations $\pi$, consider the linear program that has parameters $s(1)$, ..., $s(k)$ as variables (assuming $s(0)=0$) and its constraints require that each expression of block $B^P_{A,\pi}(s)$ ---each of which is a linear function of the variables--- is non-negative and, furthermore, $0\leq s(1) \leq ... \leq s(k)$ and $s(k) \geq 1$ to ensure non-negativity, monotonicity, and non-triviality. If the linear program is feasible for some permutation $\pi$, then the corresponding scoring function $s$ gives a sequential Thiele rule that makes $A$ a winning committee in profile $P$. Otherwise, no such rule exists. Checking feasibility can be done in polynomial time using well-known algorithms for linear programming. The next statement summarizes this discussion.

\begin{thm}\label{thm:fpt}
\SeqThiele\ parameterized by the committee size $k$ is in FPT.
\end{thm}

\subsection{Sample complexity bounds}
By adapting our analysis in Section~\ref{sec:learnability} and using blocks of linear functions to witness winning committees as discussed above, we can prove Theorem~\ref{thm:seq-thiele-pac-learnable}. The sample complexity of sequential Thiele rules is polynomial, too. Indeed, it is asymptotically equivalent to the sample complexity of ABCS rules, after replacing $|\mathcal{X}_{m,k}|$ with $k+1$. 

\begin{thm}\label{thm:seq-thiele-pac-learnable}
The class $\mathcal{F}^k_{\text{seq}}$ of sequential Thiele rules with $m$ alternatives and committee size $k$ is PAC-learnable with sample complexity $s(\delta,\epsilon)$ such that
$$s(\delta,\epsilon) \in \Omega\left(\epsilon^{-1} \left(k+\ln{(1/\delta)}\right)\right),$$
and 
$$s(\delta,\epsilon) \in \mathcal{O}\left(\epsilon^{-1} \left(k^2\cdot \ln{m}\cdot \ln{(1/\epsilon)}+\ln{(1/\delta)}\right)\right).$$
\end{thm}

Again, the proof of Theorem \ref{thm:seq-thiele-pac-learnable} relies on Theorem \ref{thm:sample-vs-graph_Natarajan} and follows after proving an upper bound on quantity $D_G(\mathcal{F}^k_{\text{seq}})$ (Lemma~\ref{lem:graph-dim-seq}) and a lower bound on quantity $D_N(\mathcal{F}^k_{\text{seq}})$ (Lemma~\ref{lem:gen-dim-lower-seq}). 

\begin{lemma}\label{lem:graph-dim-seq}
$D_G(\mathcal{F}^k_{\text{seq}}) \in \mathcal{O}(k^2 \ln{m})$.
\end{lemma}

\begin{proof}
We follow a similar approach to the proof of Lemma~\ref{lem:non_sequential_upper_bound_graph_dim}. Assume that the graph dimension of $\mathcal{F}^k_{\text{seq}}$ is $N$, with $N\geq 4k$. Thus, there are $N$ different profiles $\{P_j\}_{j\in [N]}$ and a sequential Thiele rule $g\in \mathcal{F}^k_{\text{seq}}$ such that for any subset $S \subseteq [N]$ there exists a rule $h_S \in \mathcal{F}^k_{\text{seq}}$ with the property that
\begin{itemize}
\item $h_S(P_j) = g(P_j)$ for all $j \in S$, and
\item $h_S(P_j) \neq g(P_j)$ for all $j \in [N]\setminus S$.
\end{itemize}
Again, let $\mathcal{C}^k$ be the set of all $k$-sized committees. We now argue that the set $\mathcal{L}$ of linear functions defined by the blocks $B^{P_j}_{C,\pi}(s)$ for every committee $C \in \mathcal{C}^k$, every permutation $\pi:[k]\rightarrow [k]$, and every $j\in [N]$ realize at least $2^N$ distinct sign patterns. 

Now, consider two different subsets $S$ and $S'$ of $[N]$ such that $S\not\subseteq S'$. Let $j^*$ be such that $j^*\in S\setminus S'$. Since $\mathcal{F}^k_{\text{seq}}$ G-shatters the set of profiles $\{P_j\}_{j \in [N]}$ by our assumption, there exist two rules $h_S, h_{S'} \in \mathcal{F}^k_{\text{seq}}$ such that $h_S(P_{j^*}) = g(P_{j^*})$ and $h_{S'}(P_{j^*}) \neq g(P_{j^*})$. Then, there must be a pair of committees $C,D \in \mathcal{C}^k$ such that $C \in h_{S}(P_{j^*}), D \in h_{S'}(P_{j^*})$ and either $C \notin h_{S'}(P_{j^*})$ or $D \notin h_{S}(P_{j^*})$. Thus, at least one of the following cases must happen:
\begin{itemize}
\item $C\not\in h_{S'}(P_{j^*})$. Then, for every permutation $\pi$, one of the linear functions of block $B^{P_{j^*}}_{C,\pi}(s)$ is negative for $s = h_{S'}$. In contrast, by our assumption that $C \in h_{S}(P_{j^*})$, we know that there is a permutation $\pi$ so that all linear functions of block $B^{P_{j^*}}_{C,\pi}(s)$ are non-negative for $s = h_{S}$. Hence, setting $s=h_S$ and $s=h_{S'}$ yields different sign patterns to the set of linear functions $\mathcal{L}$.
\item $D\not\in h_{S}(P_{j^*})$. Then, for every permutation $\pi$, one of the linear functions of block $B^{P_{j^*}}_{D,\pi}(s)$ is negative for $s=h_{S}$. In contrast, by our assumption that $D\in h_{S'}(P_{j^*})$, there is a permutation $\pi$ so that all linear functions of block $B^{P_{j^*}}_{D,\pi}(s)$ are non-negative for $s=h_{S'}$. Again, setting $s=h_S$ and $s=h_{S'}$ yields different sign patterns to $\mathcal{L}$.
\end{itemize} 

We now apply Theorem~\ref{thm:alon_extension_warren} to $\mathcal{L}$ with $\tau=1$ and $\ell=k$.  To bound the number $K$ of linear functions, observe that $\mathcal{L}$ contains for each of the $N$ profiles $P_1, \ldots, P_N$, every committee $C \in \mathcal{C}^k$, and each permutation $\pi$ a block with at most $k\cdot m$ functions. Since $|\mathcal{C}^k|={m \choose k}$,  we have $K\leq N\cdot {m\choose k}\cdot k!\cdot k \cdot m\leq N \cdot m^{k+1}\cdot k$. Theorem~\ref{thm:alon_extension_warren} then gives us an upper bound of $(8eNm^{k+1})^k$ different sign patterns; this must be at least as high as the number of distinct voting rules defined by the $2^N$ possible selections of a set $S\subseteq [N]$. We obtain that
\begin{align*}
2^{N/k}-8eNm^{k+1} \leq 0,
\end{align*}
which implies
\begin{align*}
\frac{k}{N}\cdot 2^{N/k}\leq 8em^{k+2}.
\end{align*}
We now apply the inequality $2^{z/2}\geq z$ for $z\geq 4$. Using $z=N/k$,  combined with our assumption that $N\geq 4k$, we get
$2^{\frac{1}{2} \cdot N/k}\leq 8em^{k+2}$,
which yields $N\leq 2k\log{\left(8em^{k+2}\right)}$, as desired.
\end{proof}

\begin{lemma}\label{lem:gen-dim-lower-seq}
$D_N(\mathcal{F}^k_{\text{seq}})\in\Omega(k)$.
\end{lemma}

\begin{proof}
Let $m=k+1$ and consider the set of alternatives $\Sigma=\{b_1, b_2, ..., b_{k-1}, a, c\}$. We define  $k-1$ profiles and a set $F$ of $k-1$ rules from $\mathcal{F}^k_{\text{seq}}$ which $N$-shatters these profiles.

For $x=2, 3, ..., k$, profile $P_x$ contains: three votes $\sigma_1^{x,i}=\sigma_2^{x,i}=\sigma_3^{x,i}=\{b_i\}$ for $i=1, 2, ..., k-1$, a vote $\sigma^x_4=\{a\}$, and a vote $\sigma^x_5=\{b_1, b_2, ..., b_{x-1},c\}$. Also, define the family of voting rules $F\subseteq \mathcal{F}^k_{\text{seq}}$ which, for every subset $S$ of $\{2, 3, ..., k\}$, contains the rule $h_S$ defined as:
\begin{align*}
h_S(0) &=0\\
h_S(1) &=1\\
h_S(x)-h_S(x-1) &= \begin{cases}
0, \quad\quad \text{if } x\in S\\
2, \quad\quad \text{if } x\in \{2, 3, ..., k\} \setminus S
\end{cases}
\end{align*}
Let us see how rule $S$ works for profile $P_x$. Initially, the winning committee is empty. Adding alternative $b_i$ to the winning committee increases the score by $4$ if $1\leq i< x$ and by $3$ if $x\leq i \leq k-1$. Adding alternatives $a$ or $c$ would increase the score by $1$. Hence, some alternative $b_i$ for $i< x$ is included in the winning committee in the first step. 

Next, we claim that all $b$-alternatives are included in the winning committee in the first $k-1$ steps. Indeed, assume that $j$ $b$-alternatives (with $1\leq j< k-1$) have been included in the winning committee. At step $j+1$, adding another $b$-alternative increases the score by at least $3$. Adding alternative $c$ would increase the score by $h_S(j+1)-h_S(j)$ if $j< x$ and by $h_S(x)-h_S(x-1)$ if $j\geq x$, i.e., by no more than $2$. Adding alternative $a$ would increase the score by just $1$. Hence, some $b$-alternative is included in the winning committee at step $j+1$, too. 

At the final $k$-th step, we are left with alternatives $a$ and $c$. Including alternative $a$ in the committee increases the score by $1$. Including alternative $c$ increases the score by $h_S(x)-h_S(x-1)$, which is either $0$ or $2$. Hence, $A=\{b_1, ..., b_{k-1},a\}$ or $C=\{b_1, ..., b_{k-1},c\}$ is the unique winning committee returned by rule $h_S$ when applied on profile $P_x$, depending on whether  $x$ belongs to $S$ or not.

By Definition~\ref{defn:shattering}, this means that the family $F$ (and, consequently, the family $\mathcal{F}^k_{\text{seq}}$) $N$-shatters the profiles $P_2$, ..., $P_{k}$ (by defining $g_1$ and $g_2$ as $g_1=h_{\{2, 3, ..., k\}}$ and $g_2=h_{\emptyset}$, respectively). By Definition~\ref{defn:graph_natarajan_dimension}, we conclude that $D_N(\mathcal{F}^k_{\text{seq}})\geq k-1$.
\end{proof}

\subsection{$\mathsf{NP}$-hardness of \SeqThiele}\label{sec:complexity_sequential}
The last statement of this section is negative and provides evidence that learning in class $\mathcal{F}^k_{\text{seq}}$ is hard as well. The proof employs a novel reduction from a structured version of \textsc{3SAT}, known to be $\mathsf{NP}$-hard~\citep{Y05}.

\begin{defn}[\textsc{2P2N-3SAT} problem]\label{defn:2P2N-3SAT_problem}
In the \textsc{2P2N-3SAT} problem, we are given a 3-CNF formula $\phi$ with $r$ variables $x_1, \ldots, x_r$ and $t$ clauses $c_1, \ldots, c_t$. Each variable appears in $\phi$ twice as a positive literal and twice as a negative literal. The goal is to decide whether there exists an assignment $\alpha$ of boolean values to the variables so that $\phi$ is satisfied.
\end{defn}

\begin{thm}\label{thm:seq-thiele-np-hard}
\SeqThiele\ is $\mathsf{NP}$-hard.
\end{thm}

\begin{proof}
We prove the theorem by presenting a polynomial-time reduction from the \textsc{2P2N-3SAT} problem.  Let $\phi$ be the given 3-CNF formula on $r$ variables and $t$ clauses.  In our reduction, we have various types of alternatives. There are eight \textit{padding alternatives} $\{p, w_1, \ldots, w_7\}$. The $2r$ literals derived from the variables of $\phi$ form another set of \textit{literal alternatives} $\{x_1, \bar{x}_1, x_2, \bar{x}_2, \ldots, x_r, \bar{x}_r\}$. The $t$ clauses of $\phi$ are represented by a set of \textit{clause alternatives} $\{c_1, \ldots, c_t\}$. There is a set of \textit{special alternatives} $\{s_1, \ldots, s_t,z\}$. In addition, the set of alternatives $\Sigma$ includes a number of dummy alternatives which we denote by $d$ followed by a subscript. Let $k=2r+t+8$ and
$$A= \{
p, w_1, \ldots, w_7,
x_1,\bar{x}_1, x_2, \bar{x}_2, \ldots, x_r, \bar{x}_r,
c_1, \ldots, c_t\}.$$
Before we define profile $P$, we introduce some notation. For any clause $c$ of $\phi$, let
$$\text{lit}(c) = \{l_1, l_2, l_3\} \subseteq \{x_1, \bar{x}_1, x_2, \bar{x}_2, \ldots, x_r, \bar{x}_r\}$$
be the set of literals appearing in $c$. Let $S$ be a set of alternatives. For some ordering of the alternatives in $S$ and any $i\in[|S|]$, we use $S(i)$ to denote the $i$-th alternative of the set $S$. The profile $P$ then again consists of three parts. 
\begin{itemize}
\item Part 1 is derived from the given 3-CNF formula $\phi$. For every $i \in [r]$, profile $P$ includes three copies of vote $\{x_i, \bar{x}_i\}$ and the votes $\{x_i, d_{x_i}\},\{\bar{x}_i,d_{\bar{x}_i}\}$. For every $i \in [t]$, there are the votes $\{c_i, l_1\}, \{c_i, l_2\}, \{c_i, l_3\}$ where $\{l_1, l_2, l_3\} = \text{lit}(c_i)$, two copies of the vote $\{c_i, s_i\}$, and a vote $\{s_i, d_{s_i}\}$. Figure~\ref{fig:NP_hardness_reduction_seq} shows a graph representation of part 1 of profile $P$.
\item Part 2 consists of the following votes: votes $\{p,z\}$ and $\{p,d_{p,j}\}$ for every $j \in [9]$, and votes $\{w_i,z\}$ and $\{w_i, d_{w_i, j}\}$ for every $i \in [7], j \in [6]$.
\item For part 3, we define $S = A \cup \{s_1, \ldots, s_t\}$ and impose an arbitrary ordering over the alternatives in $S$. Notice that $|S| = k+t$. Then, part 3 consists of the votes $\{S(1), S(2),..., S(i-1), z, S(i+1),..., S(k+t)\}$ for $i \in [k+t]$.
\end{itemize}

\begin{figure*}[t]
\centering
\begin{tikzpicture}
\begin{scope}[scale=0.75,
              every node/.style={circle,draw,thin,
                                    font=\small,
                                    minimum size=18pt,inner sep=0pt, outer sep=0pt},
              every edge/.style={draw,very thin}]
    % Upper part: 3SAT instance (origin (0,0) is at node p)
    \begin{scope}[every node/.style={circle,draw,thin,fill=lightgray,
                                    font=\normalsize,
                                    minimum size=18pt,inner sep=0pt, outer sep=0pt}]
        % Candidates from A
        \node (x1)      at (-2,3) {$x_1$};
        \node (x1_bar)  at (-1,3) {$\bar{x}_1$};
        \node (x2)      at (1, 3) {$x_2$};
        \node (x2_bar)  at (2, 3) {$\bar{x}_2$};
        \node (xr)      at (5, 3) {$x_r$};
        \node (xr_bar)  at (6, 3) {$\bar{x}_r$};
        
        \node (c1)  at (10, 3) {$c_1$};
        \node (c2)  at (12, 3) {$c_2$};
        \node (ct)  at (15, 3) {$c_t$};
        %\node (q)   at (17, 3) {$q$};
    \end{scope}
    \begin{scope}[every node/.style={circle,draw,thin,
                                    font=\normalsize,
                                    minimum size=18pt,inner sep=0pt, outer sep=0pt}]
        % Special candidates
        \node (s1)       at (10, 1)  {$s_1$};
        \node (s2)       at (12, 1)  {$s_2$};
        \node (st)       at (15, 1)  {$s_t$};
        %\node (z)        at (17, 1)  {$z$};
    \end{scope}
    \begin{scope}[every node/.style={font=\small}]
        % Dots
        \node at (3.5, 3) {$\ldots$};
        %\node at (3.5, 2) {$\ldots$};
        \node at (13.5,3) {$\ldots$};
        \node at (13.5,1) {$\ldots$};
        \node at (13.5,0) {$\ldots$};
    \end{scope}
    
    \begin{scope}[every node/.style={draw,circle,fill,black,
                                minimum size=6pt,inner sep=0pt, outer sep=0pt}]
        % dummy candidates
        \node (y1)      at (-2.50,2) {};
        \node (y1_bar)  at (-0.5,2) {};
        \node (y2)      at (0.5, 2) {};
        \node (y2_bar)  at (2.5, 2) {};
        \node (yr)      at (4.5, 2) {};
        \node (yr_bar)  at (6.5, 2) {};
        
        \node (s1_hat)  at (10, 0)  {};
        \node (s2_hat)  at (12, 0) {};
        \node (st_hat)  at (15, 0) {};
    \end{scope}
    
    % Edges upper part
    %\path [-] (x1) edge[bend right=30] (x1_bar);
    \path [-] (x1) edge[looseness=1.1,bend right=50] (x1_bar);
    \path [-] (x1) edge[looseness=1.4,bend right=70] (x1_bar);
    \path [-] (x1) edge[looseness=1.8,bend right=85] (x1_bar);
    \path [-] (x1) edge (y1);
    \path [-] (x1_bar) edge (y1_bar);
    
    %\path [-] (x2) edge[bend right=30] (x2_bar);
    \path [-] (x2) edge[looseness=1.1,bend right=50] (x2_bar);
    \path [-] (x2) edge[looseness=1.4,bend right=70] (x2_bar);
    \path [-] (x2) edge[looseness=1.8,bend right=85] (x2_bar);
    \path [-] (x2) edge (y2);
    \path [-] (x2_bar) edge (y2_bar);
    
    %\path [-] (xr) edge[bend right=30] (xr_bar);
    \path [-] (xr) edge[looseness=1.1,bend right=50] (xr_bar);
    \path [-] (xr) edge[looseness=1.4,bend right=70] (xr_bar);
    \path [-] (xr) edge[looseness=1.8,bend right=85] (xr_bar);
    \path [-] (xr) edge (yr);
    \path [-] (xr_bar) edge (yr_bar);
    
    \path [-] (c1) edge[bend left=30] (s1);
    %\path [-] (c1) edge (s1);
    \path [-] (c1) edge[bend right=30] (s1);
    
    \path [-] (c2) edge[bend left=30] (s2);
    %\path [-] (c2) edge (s2);
    \path [-] (c2) edge[bend right=30] (s2);
    
    \path [-] (ct) edge[bend left=30] (st);
    %\path [-] (ct) edge (st);
    \path [-] (ct) edge[bend right=30] (st);
    
    %\path [-] (q) edge[bend left=30] (z);
    %\path [-] (q) edge (z);
    %\path [-] (q) edge[bend right=30] (z);
    
    \path [-] (s1) edge (s1_hat);
    \path [-] (s2) edge (s2_hat);
    \path [-] (st) edge (st_hat);
    
    \path [-] (x1) edge ($(x1)!1cm!(-1.72,4)$);%(-2,3)
    \path [-] (x1) edge[bend left, out=60, in=100] (c1);
    
    \path [-] (x1_bar) edge ($(x1_bar)!1cm!(-.75,4)$);%(-1,3)
    \path [-] (x1_bar) edge ($(x1_bar)!1cm!(-.4,4)$);
    
    \path [-] (x2) edge ($(x2)!1cm!(1.25,4)$);%(1,3)
    \path [-] (x2) edge ($(x2)!1cm!(1.6,4)$);
    
    \path [-] (x2_bar) edge ($(x2_bar)!1cm!(2.4,4)$);%(2,3)
    \path [-] (x2_bar)[bend left, out=50, in=117.5] edge (c1);
    
    \path [-] (xr) edge ($(xr)!1cm!(5.25,4)$);%(5,3)
    \path [-] (xr) edge ($(xr)!1cm!(5.6,4)$);
    
    \path [-] (xr_bar) edge ($(xr_bar)!1cm!(6.4,4)$);%(6,3)
    \path [-] (xr_bar)[bend left, out=45, in=135] edge (c1);
    
    \path [-] (c2) edge ($(c2)!1cm!(12,4)$);%(12,3)
    \path [-] (c2) edge ($(c2)!1cm!(11.75,4)$);
    \path [-] (c2) edge ($(c2)!1cm!(11.4,4)$);
    
    \path [-] (ct) edge ($(ct)!1cm!(15,4)$);%(7,3)
    \path [-] (ct) edge ($(ct)!1cm!(14.75,4)$);
    \path [-] (ct) edge ($(ct)!1cm!(14.4,4)$);
\end{scope}
\end{tikzpicture}
\caption{A graph representation of part 1 of profile $P$. The vertices correspond to alternatives  and the edges to votes. Alternatives belonging to committee $A$ are highlighted in gray. In this example, we assume that $\text{lit}(c_1) = \{x_1, \bar{x}_2, \bar{x}_r\}$. The other votes containing a literal alternative and a clause alternative are only indicated as edge stubs. White vertices are the special alternatives. Dummy alternatives are shown as small black vertices without a label.}
\label{fig:NP_hardness_reduction_seq}
\end{figure*}
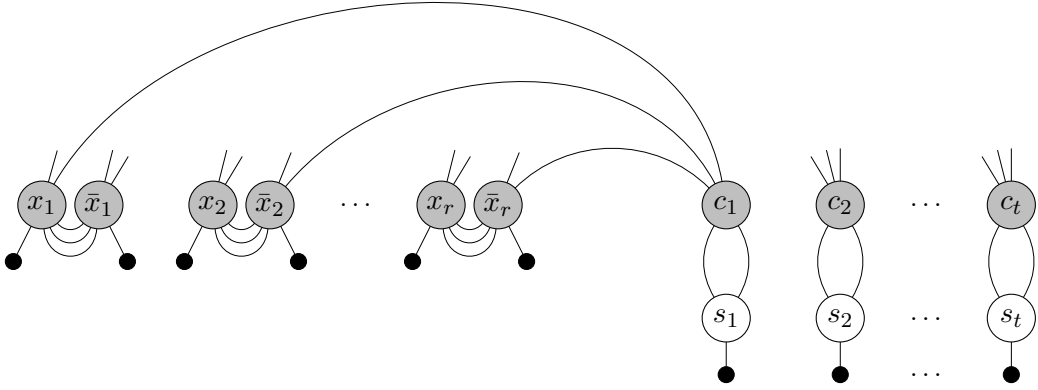

\noindent We refer to the subprofiles of $P$ in part 1, 2, and 3 as $P_1, P_2$, and $P_3$, respectively.

Throughout the proof, for the sequential Thiele rule $f$ applied to profile $P_j$, we use the notation $\Delta_f(C_i,a,P_j)$ to denote the \textit{score increase} when including alternative $a$ in the subcommittee $C_i$, i.e., 
$$\Delta_f(C_i,a,P_j)=\score_f(C_i\cup\{a\},P_j)-\score_f(C_i,P_j).$$ 
In Claims~\ref{claim:seq_p3_does_not_affect_order} and \ref{claim:seq_p3_z_higher_score}, we make two important observations about part 3 of profile $P$.

\begin{claim}\label{claim:seq_p3_does_not_affect_order}
For any rule $f \in \mathcal{F}^k_{\text{seq}}$, let $C_i\subseteq A$ be the set of winning alternatives in step $i \in \{0,\ldots,k-1\}$. For any two alternatives $a,a' \in S\setminus C_i$, it holds that
$$\Delta_f(C_i, a, P_3) = \Delta_f(C_i, a', P_3).$$
\end{claim}
\begin{proof}
For any alternative $a \in S\setminus C_i$, observe that $a$ appears in $k+t-1$ votes of $P_3$. In $i$ of these votes, $a$ appears together with $z$ and $i-1$ alternatives from $C_i$. In the remaining $k+t-1-i$ votes, $a$ appears together with all $i$ alternatives from $C_i$. Recall that $z\notin A$ and therefore $z\notin C_i\subseteq A$. We thus have that
\begin{align}\label{eqn:score_increase_a_in_A}
\Delta_f(C_i, a, P_3)
&= (k+t-1-i) (f(i+1) - f(i)) + i(f(i)-f(i-1)).
\end{align}
Notice that this statement is well-defined also for $i=0$ since the term that contains $f(i-1)=f(-1)$ vanishes in this case. Thus, any two alternatives $a,a' \in S\setminus C_i$ receive the same additional score from subprofile $P_3$.
\end{proof}

\begin{claim}\label{claim:seq_p3_z_higher_score}
For any rule $f \in \mathcal{F}^k_{\text{seq}}$, let $C_i\subseteq A$ be the set of winning alternatives in step $i \in \{0,\ldots,k-1\}$. For any alternative $a \in S\setminus C_i$, it holds that
$$\Delta_f(C_i, z, P_3) \geq \Delta_f(C_i, a, P_3),$$
with strict inequality if $f(i+1) - f(i) > 0$.
\end{claim}
\begin{proof}
Again, recall that $z\notin A$ such that $z\notin C_i\subseteq A$. Alternative $z$ appears in all $k+t$ votes of $P_3$. In $i$ of these votes, $z$ appears together with $i-1$ alternatives from $C_i$. In the remaining $k+t-i$ votes, $z$ appears together with all $i$ alternatives from $C_i$. The claim then follows from the observation that 
\begin{align*}
\Delta_f(C_i, z, P_3)
= (k+t-i) (f(i+1) - f(i)) + i(f(i)-f(i-1))
= \Delta_f(C_i, a, P_3) + f(i+1) - f(i),
\end{align*}
where we used~(\ref{eqn:score_increase_a_in_A}) to obtain the second equality.
\end{proof}

We proceed by narrowing down the choices of parameters that may make $A$ a winning committee in $P$.
\begin{lemma}\label{lem:seq_A_winning_only_if_f1_eq_f2_greater_zero}
For any non-trivial rule $f \in \mathcal{F}^k_{\text{seq}}$, committee $A$ is winning in $P$ only if $f(1) = f(2) > 0$.
\end{lemma}

\begin{proof}
We distinguish between three cases where the condition $f(1) = f(2) > 0$ is not true, in each of which $A$ is not winning. First, assume that $f(1) = f(2) = 0$. Since $f$ is non-trivial, there is some $i \in \{2,\ldots,k-1\}$ such that $f(i+1) - f(i) > 0$. Let $i$ be minimal such that this is the case and let $C_i\subseteq A$ be the winning committee in step $i$. Since $i\geq 2$, subprofiles $P_1, P_2$ do not grant additional score to any alternative $a \in A\setminus C_i$. From Claim~\ref{claim:seq_p3_z_higher_score}, it then follows that alternative $z$ has strictly higher score increase than any alternative $a \in A\setminus C_i$. Thus, alternative $z$ is included in the winning committee such that $A$ is not winning.

Now, assume that $f(1) = 0$ and $f(2) > 0$. Note that the padding alternatives $p, w_1, \ldots, w_7$ belong to committee $A$. Consider the first step $i$ in which an alternative from $\{p,w_1, ..., w_7\}$ is included in the winning subcommittee $C_i\subseteq A$. Since no two alternatives from $\{p,w_1, ..., w_7\}$ belong together in a vote of $P_2$ and since $f(1)=0$, we have
\begin{align*}
\Delta_f(C_i,a,P) &=\Delta_f(C_i,a,P_3)
\end{align*}
for every $a\in \{p,w_1, ..., w_7\}\setminus C_i$. Consider the alternative $z$. We observe that $z$ does not appear in any vote of $P_1$ and shares exactly one vote with $a$ among the votes of $P_2$. Hence,
\begin{align*}
\Delta_f(C_i, z, P)
= \Delta_f(C_i, z, P_2) + \Delta_f(C_i, z, P_3)
&= f(2) + \Delta_f(C_i, z, P_3)\\
&\geq f(2) + \Delta_f(C_i, a, P_3)\\
&= f(2) + \Delta_f(C_i, a, P)\\
&> \Delta_f(C_i, a, P).
\end{align*}
Here, the first inequality is due to Claim~\ref{claim:seq_p3_z_higher_score}, and the second (strict) inequality follows from the assumption that $f(2) > 0$. Hence, alternative $z$ is included in the winning committee before any alternative $a \in  \{p, w_1, \ldots, w_5\}\setminus C_i$.

Lastly, assume that $f(2) > f(1) > 0$. In this case, the winning committee in step $1$ is $C_1 = \{p\}$, since alternative $p$ appears in the most votes from profile $P$. More specifically, $p$ appears in 10 votes of $P_2$ and $k+t-1$ votes of $P_3$ (i.e., $9+k+t$ votes in total) while all other alternatives from $A$ appear in at most 7 votes among $P_1 \cup P_2$ and $k+t-1$ votes of $P_3$ (i.e., at most $6+k+t$ votes in total). Alternative $z$ appears in 8 votes of $P_2$, and $k+t$ votes of $P_3$ (i.e., $8+k+t$ votes in total). By our construction, the remaining alternatives clearly appear in strictly less votes than $p$.

Consider an alternative $a\in A\setminus\{p\}$. Again, observe that, in profiles $P_1$ and $P_2$, alternative $a$ appears in at most 7 votes, and none of these votes contain $p$. Furthermore, alternative $a$ appears in $k+t-1$ votes of $P_3$, $k+t-2$ of which also contain $p$. Hence, 
\begin{align}\label{eqn:singleton_subcommittee_p}
\Delta_f(\{p\},a,P)
&\leq 8 f(1) + (k+t-2)(f(2)-f(1)).
\end{align}
On the other hand, alternative $z$ appears in 8 votes of $P_2$, and in all $k+t$ votes of $P_3$. Among these votes, $p$ belongs to one vote of $P_2$, and to $k+t-1$ votes of $P_3$. Hence,
\begin{align*}
\Delta_f(\{p\},z,P)
&=7f(1) + f(2) - f(1) + f(1) + (k+t-1)(f(2)-f(1))\\
&= 8f(1) + (k+t-2)(f(2)-f(1)) + 2(f(2)-f(1))\\
&\geq \Delta_f(\{p\},a,P) + 2(f(2)-f(1))\\
&> \Delta_f(\{p\},a,P),\\
\end{align*}
where we used~(\ref{eqn:singleton_subcommittee_p}) to obtain the first inequality, and the second (strict) inequality is due to the assumption that $f(2) > f(1) > 0$. Thus, in the second step, alternative $z$ is included in the winning committee before any alternative $a \in A\setminus \{p\}$.
\end{proof}

We complete the proof assuming ---without loss of generality--- that $f(1) = f(2) = 1$. For any such rule, we make a simple, yet important observations which follows from Claim~\ref{claim:seq_p3_does_not_affect_order}.

\begin{obs}\label{obs:order_inclusion_depends_only_on_P1_P2}
Consider any two alternatives $a,a'\in S$ and let rule $f$ be such that $f(1)=f(2)=1$. At any step during the evaluation of $f$ on profile $P$, whether alternative $a$ has larger score increase than alternative $a'$ depends only on the votes of $P_1$ and $P_2$.
\end{obs}

Thus, in order for $A$ to be winning, alternatives $p, w_1, \ldots, w_7$ are included in the winning committee first. Then, for every variable $x_i$ in $\phi$, either alternative $x_i$ or $\bar{x}_i$ (exclusively) is included. After that, any clause alternative $c$ is selected as long as there is at least one literal alternative $l\in\text{lit}(c)$ such that $l$ is not yet included in the winning committee.\footnote{There is an intuitive way of keeping track of the score increase that the alternatives earn from profile $P_1$ using the graph representation in Figure~\ref{fig:NP_hardness_reduction_seq}. For $f(1) = f(2) = 1$, including an alternative in the winning committee is equivalent to removing all edges incident with this alternative's node from the graph. At any step, the score increase that an alternative receives from $P_1$ is the number of edges incident with this alternative's node in the respective step.}

\begin{lemma}\label{lem:seq_z_before_some_c_if_not_satisfyable}
If there is no assignment to the variables $x_1, \ldots, x_r$ satisfying $\phi$, then there is some $j^* \in [t]$ such that alternative $s_{j^*}$ is included in the winning committee before alternative $c_{j^*}$.
\end{lemma}
\begin{proof}
Let $\widetilde{A}$ be the winning committee after the inclusion of alternatives $p, w_1,\ldots,w_5$ and either $x_i$ or $\bar{x}_i$ (exclusively) for every $i \in [r]$. Since there is no assignment to the variables that satisfies $\phi$, there must now be a clause alternative $c_{j^*}$ for which all literal alternatives corresponding to literals in $\text{lit}(c_{j^*})$ are already included in the winning committee. Otherwise, the ordering in which the $r$ literal alternatives have been included in the winning committee would correspond to an assignment $\alpha$ that satisfies $\phi$. That is, we could pick $\alpha$ such that
$$\alpha(l) =
\begin{cases}
  	0 &\text{if } l \in \widetilde{A},\\
  	1 &\text{otherwise,}
  \end{cases}$$
for every $l \in \{x_1, \bar{x}_1, \ldots, x_r, \bar{x}_r\}$.

In order for $A$ to be winning on profile $P$, $c_{j^*}$ needs to included in the winning subcommittee in some step. However, notice that for any $C_i\subseteq A$, it holds that
$$\Delta_f(C_i,s_{j^*},P_1) = 3 > 2 = \Delta_f(C_i,c_{j^*},P_1).$$
Outside of $P_1$, the alternatives $c_{j^*},s_{j^*}$ appear only in votes of $P_3$. But, from Claim~\ref{claim:seq_p3_does_not_affect_order}, it follows that $\Delta_f(C_i,c_{j^*},P_3) = \Delta_f(C_i,s_{j^*},P_3)$. Hence, $\Delta_f(C_i,s_{j^*},P) > \Delta_f(C_i,c_{j^*},P)$ such that $s_{j^*}$ is included in the winning committee before $c_{j^*}$.
\end{proof}

So far, we have shown that committee $A$ is not winning on profile $P$ under any non-trivial rule $f \in \mathcal{F}^k_{\text{seq}}$ if there exists no assignment that satisfies $\phi$. Now, assume that there is an assignment $\alpha$ to the variables $x_1, \ldots, x_r$ that satisfies $\phi$. We rename the literals according to their values under the assignment $\alpha$ such that
$$l^1_i = \begin{cases}
       	x_i & \text{if } \alpha(x_i) = 1\\
       	\bar{x}_i & \text{otherwise,}
       \end{cases}$$
and
$$l^0_i = \begin{cases}
       	x_i & \text{if } \alpha(x_i) = 0\\
       	\bar{x}_i & \text{otherwise.}
       \end{cases}$$
Let $f$ be such that $f(1) = f(2) = \ldots = f(k) = 1$. We conclude the proof by showing that there is an ordering of the alternatives in $A$ such that the sequential Thiele rule specified by $f$ returns $A$ as the winning committee. Using Observation~\ref{obs:order_inclusion_depends_only_on_P1_P2} and, with respect to the score increase of alternative $z$, also Claim~\ref{claim:seq_p3_z_higher_score}, the feasibility of the following sequence of inclusions can easily be verified.

First, the alternatives $p, w_1, \ldots, w_7$ and the literal alternatives $l^0_1, \ldots, l^0_r$ are added to the committee. Conversely, none of the alternatives $l^1_1, \ldots, l^1_r$ have been included in the winning committee yet. Since $\alpha$ is a satisfying assignment, each clause alternative $c_i$ for $i \in [t]$ is included in at least one approval vote $\{l,c_i\}$ where $l \in \{l^1_1, \ldots, l^1_r\}$. Thus, every clause alternative $c_i$ has a score increase of at least $3$ which is at least as high as the score increase of the corresponding special alternative $s_i$ and of any remaining literal alternative. This implies that there is an ordering of the clause alternatives such that each of these alternatives is included in the winning committee. At this point, all the remaining alternatives have a score increase of at most $1$. All alternatives from $A$ that have not yet been selected for the winning committee have a score increase of exactly $1$. Hence, the alternatives $l^1_1, \ldots, l^1_r$ can be included in the winning committee as well. Thereby, $A$ is indeed winning under rule $f$ on profile $P$.
\end{proof}

We remark that by considering only part 1 of profile $P$, our reduction yields the $\mathsf{NP}$-hardness of winner verification for the sequential CC rule (see Appendix~\ref{sec:app} for a detailed statement and proof).

\section{Concluding Remarks}\label{sec:conclusion}
We studied complexity aspects of learning ABCS and sequential Thiele rules. In a nutshell, our results suggest that learning from these classes is feasible in the PAC learning framework but ---in a worst-case sense--- only in computational inefficient ways. We believe that our techniques for assessing PAC learnability can be extended to other rules. \citet{FSST19} define a hierarchy of classes of ranking-based multiwinner voting rules that are specified using scoring functions. These are natural candidates for extending our analysis. We also remark that, en route to proving hardness of \ABCS\ and \SeqThiele, parts of our reductions show hardness of winner verification for the CC and the sequential CC rule.  Formal statements and proofs of these two byproduct results appear in Appendix~\ref{sec:app}.

\paragraph{Acknowledgements} We thank Piotr Faliszewski for his insights on the complexity of sequential Thiele rules and for pointers to the literature.  We also thank Allan Gr{\o}nlund for helpful discussions on the multiclass fundamental theorem, Chris Schwiegelshohn for discussion on the graph dimension, Krzysztof Sornat for comments on an earlier version of this paper, and an anonymous reviewer for spotting a subtle error in an earlier version of the proof of Theorem~\ref{thm:seq-thiele-np-hard}.

\bibliography{learnmvr.arxiv}

\appendix

\section{Hardness of Winner Verification for the CC Rule}\label{sec:app}
We introduced the decision problems \ABCS\ and \SeqThiele\ in Section~\ref{sec:prelim}. A more restricted variant of these problems is known under the term {\em winner verification}. In addition to the inputs defined for our problems, that is, a profile $P$ and a $k$-sized committee $C$, we are now also given a specific rule $f$ from $\mathcal{F}^{m,k}$ or $\mathcal{F}^k_{\mathit{seq}}$. The goal of the winner verification problem is to decide whether committee $C$ is winning on profile $P$ under the given rule $f$. Our hardness proofs for \ABCS\ and \SeqThiele\ produce as a byproduct corresponding hardness proofs for winner verification under the CC rule, and its sequential counterpart, the sequential CC rule. Recall that we introduced the CC rule in Section~\ref{sec:prelim} as an example of a Thiele rule, a rule whose scoring function does not depend on the $y$ parameter and can therefore be treated as univariate. The CC rule is specified by the univariate scoring function $\ccrule$ where $\ccrule(0) = 0$, and $\ccrule(x) = 1$ for $x>0$. The sequential CC rule is the sequential Thiele rule that uses $\ccrule$ as its scoring function.

In the following, we show $\mathsf{coW}[1]$-hardness of winner verification for the CC rule (Theorem~\ref{thm:byproduct_result-1}) and thereby strengthen a recent result by~\citet{SDM20}. We then proceed to show $\mathsf{NP}$-hardness of winner verification for the sequential CC rule  (Theorem~\ref{thm:byproduct_result-2}). Both proofs rely on simplifications to our constructions in the hardness proofs of \ABCS\ (Theorem~\ref{thm:ABCS_coW1_hard}) and \SeqThiele\ (Theorem~\ref{thm:seq-thiele-np-hard}), respectively. 

\begin{thm}\label{thm:byproduct_result-1}
Winner verification for the CC rule parameterized by the committee size $k$ is $\mathsf{coW}[1]$-hard. 
\end{thm}

\begin{proof}
We reduce from \textsc{IndependentSet} (Definition~\ref{defn:independent_set}), proceeding along the same lines as the proof of Theorem~\ref{thm:ABCS_coW1_hard}. We set $k = K$ and define the graph $G' = (V', E')$ in the same way as before. In particular, every node $v \in V'$ has degree $\Delta \geq 2$. Without loss of generality, we can assume that $V'$ is the set of positive integers in $[r]$. The set of alternatives $\Sigma$ consists of alternatives $a_i$ and $b_i$ for every vertex $i\in V'$, and an additional alternative $c$. Let $A=\{a_1, a_2, ..., a_k\}$. The profile $P$ consists of two parts: 
\begin{itemize}
\item Part 1 consists of a vote $\{b_i, b_j\}$ for every edge $(i,j)\in E'$. 
\item Part 2 consists of $k\Delta-1$ copies of each of the following votes: vote $\{a_i,b_j\}$ for every $i,j\in [r]$, and a vote $\{a_1, c\}$.
\end{itemize}
We use $P_1$, and $P_2$ to denote the two subprofiles of votes in parts 1 and 2, respectively.

Our first claim gives the score of committee $A$ in profile $P$ under the CC rule. 
\begin{claim}\label{claim:score-of-A_hardness_CC}
It holds that $\score_{\ccrule}(A,P) = (k\Delta - 1)(kr + 1)$.
\end{claim}
\begin{proof}
Committee $A$ gets one point from the $k\Delta-1$ copies of vote $\{a_i,b_j\}$ for $i\in [k]$ and $j\in [r]$, and from the $k\Delta-1$ copies of vote $\{a_1,c\}$.
\end{proof}
Now, consider a committee $B$ and let $t$ be the number of its alternatives from $\{b_1, ..., b_r\}$, and $\lambda$, and $\nu$ be binary variables indicating whether alternative $c$, and $a_1$ belong to $B$, respectively.
\begin{claim}\label{claim:score-of-B_hardness_CC}
If $t<k$, $\score_{\ccrule}(B,P) \leq (k\Delta-1)(kr+1)$.
\end{claim}
\begin{proof}
First, observe that $B$ gets at most $\Delta$ points for each of the $t < k$ alternatives from set $\{b_1, ..., b_r\}$ that $B$ contains. Hence, 
\begin{align}\label{eq:score-of-B-in-P1_hardness_CC}
\score_{\ccrule}(B,P_1) &\leq t\Delta \leq (k\Delta-1)\one\{t>0\}.
\end{align}
To bound $\score_{\ccrule}(B,P_2)$, observe that $B$ gets a point for each of the $k\Delta-1$ copies of
\begin{itemize}
\item vote $\{a_i,b_j\}$ for $i,j\in [r]$ with either $a_i\in B$ or (not exclusively) $b_j\in B$. The total number of these votes is $(k\Delta-1)(tr+(k-t-\lambda)r-t(k-t-\lambda))$. 
\item vote $\{a_1,c\}$ if at least one of $a_1$ and $c$ belong to $B$. The number of these votes is $(k\Delta-1)(\lambda+\nu-\lambda\nu)$.
\end{itemize}
Summing these numbers of votes, we get
\begin{align}\label{eq:score-of-B-in-P2_hardness_CC}
\score_{\ccrule}(B,P_2) & \leq (k\Delta-1)(kr -\lambda (r-1) - t(k-t-\lambda) + (1-\lambda)\nu).
\end{align}
Combining inequalities~(\ref{eq:score-of-B-in-P1_hardness_CC}) and~(\ref{eq:score-of-B-in-P2_hardness_CC}), we obtain
\begin{align}\label{eq:score-of-B_hardness_CC}
\score_{\ccrule}(B,P) & \leq (k\Delta-1)(\one\{t>0\}+kr -\lambda (r-1) - t(k-t-\lambda) + (1-\lambda)\nu).
\end{align}
We now distinguish between three cases:
\begin{itemize}
\item If $k-t-\lambda=0$, then, since $k \geq 2$ and $t < k$, it must hold that $ t = k-1 > 0$ implying $\one\{t>0\}$ equals 1, and $\lambda = 1$.
\item If $k-t-\lambda > 0$ and $t=0$, then $\one\{t>0\}=0$.
\item Lastly, if $k-t-\lambda > 0$ and $t>0$, we get $\one\{t>0\}=1$.
\end{itemize}
The claim for $\score_{\ccrule}(B,P)$ trivially follows from inequality~(\ref{eq:score-of-B_hardness_CC}) for all the three cases above.
\end{proof}
By Claims~\ref{claim:score-of-A_hardness_CC} and~\ref{claim:score-of-B_hardness_CC}, if committee $B$ has score higher than $\score_{\ccrule}(A,P)$, then it must be that $t=k$. We conclude the proof by reasoning about $\score_{\ccrule}(B,P)$ in this case.
\begin{claim}\label{claim:score-of-B-in-P2_hardness_CC}
Let $B$ be a committee with $t=k$. Then, $\score_{\ccrule}(B,P_2)=(k\Delta-1)kr$.
\end{claim}
\begin{proof}
In part 2 of the profile, committee $B$ gets one point from the $k\Delta-1$ copies of vote $\{a_i,b_j\}$ for $i\in [r]$ and $b_j\in B$.
\end{proof}
Note that we defined subprofile $P_1$ in exactly the same way as we did for the proof of Theorem~\ref{thm:ABCS_coW1_hard}. We also have that $\ccrule(1) = \ccrule(2) = 1$ by definition. Thus, Lemma~\ref{lem:score-of-B-in-P1} is also applicable in the context of this proof. We state the lemma here again as a corollary.
\begin{cor}\label{cor:score-of-B-in-P1_hardness_CC}
Consider any committee $B$ with $t=k$. If $G$ has no independent set of size $k$, then $\score_{\ccrule}(B,P_1)\leq k\Delta-1$.
\end{cor}
By Claim~\ref{claim:score-of-B-in-P2_hardness_CC} and Corollary~\ref{cor:score-of-B-in-P1_hardness_CC}, we obtain that if $G$ has no independent set of size $k$, then $\score_{\ccrule}(B,P)\leq (k\Delta-1)(kr+1)$. Thus, by Claim~\ref{claim:score-of-A_hardness_CC}, $A$ is a winning committee in this case. 

Now, assume that $G$ has an independent set of size $k$. This implies that $G'$ has an independent set $S$ of $k$ vertices of degree $\Delta$. Now, consider the committee $B$ consisting of the alternatives that correspond to the vertices of $S$. As the number of edges that are incident to vertices of $S$ is $k\Delta$,  we have that $\score_{\ccrule}(B,P_1)=k\Delta$ as well. Then, by Claims~\ref{claim:score-of-A_hardness_CC} and~\ref{claim:score-of-B-in-P2_hardness_CC}, we have $\score_{\ccrule}(B,P)=1+(k\Delta-1)(kr+1)>\score_{\ccrule}(A,P)$, indicating that $A$ is not winning. The proof of correctness of our reduction is now complete.
\end{proof}

\begin{thm}\label{thm:byproduct_result-2}
Winner verification for the sequential CC rule is $\mathsf{NP}$-hard.
\end{thm}

\begin{proof}
We reduce from \textsc{2P2N-3SAT} (Definition~\ref{defn:2P2N-3SAT_problem}). Let $\phi$ be the given 3-CNF formula on $r$ variables and $t$ clauses. In our reduction, we define different types of alternatives based on $\phi$. The $2r$ literals derived from the variables of $\phi$ form a set of \textit{literal alternatives} $\{x_1, \bar{x}_1, x_2, \bar{x}_2, \ldots, x_r, \bar{x}_r\}$. The $t$ clauses of $\phi$ are represented by a set of \textit{clause alternatives} $\{c_1, \ldots, c_t\}$. There is a set of \textit{special alternatives} $\{s_1, \ldots, s_t\}$. In addition, for the literal alternatives and the special candidates, there exists a corresponding set of \textit{dummy alternatives} $\{d_{x_1},d_{\bar{x}_1}, \ldots, d_{x_r},d_{\bar{x}_r},d_{s_1}, \ldots, d_{s_t}\}$. Let
$$A= \{
x_1,\bar{x}_1, x_2, \bar{x}_2, \ldots, x_r, \bar{x}_r,
c_1, \ldots, c_t\}.$$
Recall that, for any clause $c$ of $\phi$,
$$\text{lit}(c) = \{l_1, l_2, l_3\} \subseteq \{x_1, \bar{x}_1, x_2, \bar{x}_2, \ldots, x_r, \bar{x}_r\}$$
is the set of literals appearing in $c$. The profile $P$ is then defined in the exact same way as subprofile $P_1$ in the proof of Theorem~\ref{thm:seq-thiele-np-hard}. For every $i \in [r]$, profile $P$ includes three copies of vote $\{x_i, \bar{x}_i\}$ and the votes $\{x_i, d_{x_i}\}$ and $\{\bar{x}_i,d_{\bar{x}_i}\}$. For every $i \in [t]$, there are the votes $\{c_i, l_1\}, \{c_i, l_2\}, \{c_i, l_3\}$ where $\{l_1, l_2, l_3\} = \text{lit}(c_i)$, two copies of the vote $\{c_i, s_i\}$, and a vote $\{s_i, d_{s_i}\}$. The reader can inspect Figure~\ref{fig:NP_hardness_reduction_seq} for a graph representation of the profile.

\begin{claim}\label{claim:seq_CC_dummy_before_c_if_not_satisfyable}
If there is no assignment to the variables $x_1, \ldots, x_r$ satisfying $\phi$, then there is some $j^* \in [t]$ such that alternative $s_{j^*}$ is included in the winning committee before alternative $c_{j^*}$.
\end{claim}
\begin{proof}
Observe that under the sequential CC rule, for every $i \in [r]$, either $x_i$ or $\bar{x}_i$ (exclusively) is included in the winning committee in the first $r$ rounds.
Let $C_r$ be the winning committee after these rounds. Since there is no assignment to the variables that satisfies $\phi$, there must now be a clause alternative $c_{j^*}$ for which all literal alternatives corresponding to literals in $\text{lit}(c_{j^*})$ are already included in the winning committee. Otherwise, the ordering in which the $r$ literal alternatives have been included in the winning committee would correspond to an assignment $\alpha$ that satisfies $\phi$. That is, we could pick $\alpha$ such that
$$\alpha(l) =
\begin{cases}
  	0 &\text{if } l \in C_r,\\
  	1 &\text{otherwise,}
  \end{cases}$$
for every $l \in \{x_1, \bar{x}_1, \ldots, x_r, \bar{x}_r\}$.

In order for $A$ to become the winning committee, there must be a step $i$ where $c_{j^*}$ is included in the winning committee. However, observe that
$$\Delta_{\ccrule}(C_i,s_{j^*},P) = 3 > 2 = \Delta_{\ccrule}(C_i,c_{j^*},P),$$
where $\Delta_{\ccrule}(C_i,a,P)$ denotes the \textit{score increase} when including alternative $a$ in the subcommittee $C_i$ on a profile $P$. Thus, alternative $s_{j^*}$ is included in the winning committee before the alternative $c_{j^*}$.
\end{proof}

So far, we have shown that committee $A$ is not winning on profile $P$ under the sequential CC rule if there exists no assignment that satisfies $\phi$. Now, assume that there is an assignment $\alpha$ to the variables $x_1, \ldots, x_r$ that satisfies $\phi$. We rename the literals according to their values under the assignment $\alpha$ such that
$$l^1_i = \begin{cases}
       	x_i & \text{if } \alpha(x_i) = 1\\
       	\bar{x}_i & \text{otherwise,}
       \end{cases}$$
and
$$l^0_i = \begin{cases}
       	x_i & \text{if } \alpha(x_i) = 0\\
       	\bar{x}_i & \text{otherwise.}
       \end{cases}$$
We conclude the proof by showing that there is an ordering of the alternatives in $A$ such that the sequential CC rule returns $A$ as the winning committee. The 	feasibility of the following sequence of inclusions can easily be verified.

First, the literal alternatives $l^0_1, \ldots, l^0_r$ are added to the committee. Conversely, none of the alternatives $l^1_1, \ldots, l^1_r$ have been included in the winning committee yet. Since $\alpha$ is a satisfying assignment, each clause alternative $c_i$ for $i \in [t]$ is included in at least one vote $\{l,c_i\}$ where $l \in \{l^1_1, \ldots, l^1_r\}$. Thus, every clause alternative $c_i$ has a score increase of at least $3$. This score increase is at least as high as the score increase of the corresponding special alternative $s_i$ and of any remaining literal alternative. This implies that there is an ordering of the clause alternatives such that each of these alternatives is included in the winning committee. At this point, all the remaining alternatives have a score increase of at most $1$. All alternatives from $A$ that have not yet been selected for the winning committee have a score increase of exactly $1$. Hence, the alternatives $l^1_1, \ldots, l^1_r$ can be included in the winning committee as well. Thereby, $A$ is indeed winning under rule $f$ on profile $P$.
\end{proof}

\end{document}